\documentclass[12pt, draftclsnofoot, onecolumn]{IEEEtran}
\usepackage{amsmath,amssymb,amsfonts}
\usepackage{amsthm}
\usepackage{hyperref}
\usepackage[utf8]{inputenc}
\usepackage[english]{babel}
\addto\captionsenglish{}
\newtheorem{theorem}{Theorem}
\newtheorem{corollary}{Corollary}
\newtheorem{lemma}{Lemma}

\usepackage{multirow}
\usepackage{hhline}
\usepackage[ruled,vlined,linesnumbered]{algorithm2e}
\usepackage{subfigure}
\usepackage[noadjust]{cite}
\usepackage[scaled=0.92]{helvet}
\fontfamily{Name_OF_Font_Family}\selectfont
\usepackage{mathtools}

\makeatletter
\def\endthebibliography{%
	\def\@noitemerr{\@latex@warning{Empty `thebibliography' environment}}%
	\endlist
}
\makeatother

\usepackage{ifpdf}

\ifCLASSINFOpdf
   \usepackage[pdftex]{graphicx}
   \graphicspath{{figures/}}
   \DeclareGraphicsExtensions{.pdf,.jpeg,.png}
\else
   \usepackage[dvips]{graphicx}
   \graphicspath{{../eps/}}
   \DeclareGraphicsExtensions{.eps}
\fi
\usepackage{algorithmic}

\usepackage{array}

\begin{document}
%
\title{Fast Performance Evaluation of Linear Block Codes over Memoryless Continuous Channels}
%
%
%

\author{Jinzhe~Pan
        and~Wai~Ho~Mow,~\IEEEmembership{Senior~Member,~IEEE}
}

\maketitle

\begin{abstract}
	There are rising scenarios in communication systems, where the noises exhibit impulsive behavior and are not adequate to be modeled as the Gaussian distribution. The generalized Gaussian distribution instead is an effective model to describe real-world systems with impulsive noises. 
	In this paper, the problem of efficiently evaluating the error performance of linear block codes over an additive white generalized Gaussian noise (AWGGN) channel is considered. 
	The Monte Carlo (MC) simulation is a widely used but inefficient performance evaluation method, especially in the low error probability regime.
	As a variance-reduction technique, importance sampling (IS) can significantly reduce the sample size needed for reliable estimation based on a well-designed IS distribution.
	By deriving the optimal IS distribution on the one-dimensional space mapped from the observation space, we present a general framework to designing IS estimators for memoryless continuous channels.
	Specifically, for the AWGGN channel, we propose an $L_p$-norm-based minimum-variance IS estimator.
	As an efficiency measure, the asymptotic IS gain of the proposed estimator is derived in a multiple integral form as the signal-to-noise ratio tends to infinity.
	Specifically, for the Laplace and Gaussian noises, the gains can be derived in a one-dimensional integral form, which makes the numerical calculation affordable. 
	In addition, by limiting the use of the union bound to an optimized $L_1$-norm sphere, we derive the sphere bound for the additive white Laplace noise channel.
	Simulation results verify the accuracy of the derived IS gain in predicting the efficiency of the proposed IS estimator.
\end{abstract}


%
\IEEEpeerreviewmaketitle

\section{Introduction}
Monte Carlo (MC) simulation is a commonly used tool to numerically estimate the error performance of coded systems, for many of which the analytical solutions are mathematically intractable. 
However, it requires a sufficiently large number of samples to provide reliable estimates of very low word error rate (WER) or bit error rate (BER), which are usually of practical interest.
For instance, some applications for ultra-reliable low-latency communication (URLLC), such as factory automation and tele-surgery, require the WER as low as $10^{-9}$ \cite{shirvanimoghaddam2019short}. To achieve a reliable estimation, at least 100 times more simulation runs than the inverse of the target WER are required empirically.
Therefore, efficient error performance evaluation methods for coded systems are highly desirable.

As a variance-reduction technique, importance sampling (IS) is introduced for probability estimation of rare events, particularly suitable for error performance evaluation in the communication area. Various IS methods are designed to accelerate the simulation for coded systems \cite{jeruchim1984techniques,smith1997quick}. The key to the efficiency of the methods is in the choice of the proposal IS distribution. Although the global optimal IS distribution is already derived in \cite{smith1997quick}, it depends on the parameters to be estimated itself and thus is impractical.
Many IS methods use a mixture of components to approximate the optimal IS distribution, such as the multiple IS methods \cite{elvira2019generalized, elvira2021multiple}.
Usually, these methods involve an iterative adaptation of the IS distribution based on gradually optimizing the parameters of the mixture during the simulation \cite{kurtz2013cross, bugallo2017adaptive}. 
The adaptive IS is used in \cite{han2020low} for estimating the error probability of non-orthogonal multiple access (NOMA) systems. In \cite{jesep2020importance}, a nested IS method that estimates the random-coding error probability of coded-modulation systems is presented. 
There also exist methods based on nonparametric IS distributions, such as the dual adaptive IS method \cite{holzlohner2005evaluation} and the adaptive histogram-shaping MC method \cite{yu2020fast}.

The IS methods in the current literature are mainly designed based on the assumption of Gaussian noise.
Nevertheless, there are a growing number of scenarios, such as urban, indoor, and underwater, where the noises exhibit impulsive nature and the Gaussian noise assumption does not hold. 
For instance, the impulsive noises widely exist in the industrial internet-of-things \cite{chen2020robust}, smart grid and smart home presented environments and have a severe impact on the systems like NOMA \cite{selim2020noma,selim2020effect} and power line communication \cite{mohan2019secrecy, karakucs2020modelling, bai2018impulsive}.
The underwater acoustic channel in a shallow water environment can also be characterized by the presence of impulsive noises \cite{banerjee2013underwater,wang2020new}.
The generalized Gaussian distribution (GGD) is shown as an effective model for describing impulsive noises and is widely adopted in the literature. For example, it is shown in \cite{banerjee2013underwater} that the underwater acoustic channel can be well modeled by the GGD. Specifically, the cases of ship transit noise and sea surface agitation noise can be described by the GGD with shape parameters $p=2.8$ and $p=1.6$, respectively. 
In \cite{beaulieu2009designing}, the time-hopping ultra-wideband interference can be described by the GGD with shape parameter $p \le 1$ for moderate to high signal-to-noise ratio (SNR). 

Systems based on the assumption of generalized Gaussian noise arise in recent works.
In \cite{bariah2020non}, the performance of NOMA systems with additive white generalized Gaussian noise (AWGGN) is investigated.
However, there are few works on studying efficient performance evaluation methods for the AWGGN channel in the current literature.
In \cite{naseri2020fast}, a weighted counting method is presented, where the Laplace distribution is chosen as the IS distribution. The method focuses on efficient sample generation, while the bottleneck of the time cost for the MC simulation is the decoding process. The fewer the number of samples needed, the less time the simulation takes. Hence, the efficiency of the method can be further improved if a well-chosen IS distribution that minimizes the sample size required for the target reliability is adopted.
The IS gain defined in \cite{xia2003importance} is an efficiency measure of the IS method compared to the MC simulation. To the best of the authors' knowledge, none of the aforementioned works provide analytical IS gain that can predict the performance of their IS methods. It makes the technique more attractive if the users can predict how much time they can save once the IS method is adopted for their applications before the simulation.

In this paper, we develop an efficient performance evaluation method for coded systems over the AWGGN channel. 
We find that the conditional pairwise error probability (PEP) is present in the variance of the IS estimator and takes an essential role in the efficiency analysis of the proposed method.
The conditional PEP also widely appears and is of crucial significance in the derivation of error performance bounds.
For example, Gallager's first bounding technique (GFBT) improves the union bound by limiting the use of the union bound to a so-called Gallager region and avoiding over-estimating the error probability outside the region \cite{sason2006performance,divsalar1999simple,agrell1996voronoi}.
Consequently, the PEP conditioned on the given Gallager region is required in the derivation of the bound. This situation happens to many GFBT-based bounds such as the sphere bound \cite{herzberg1994techniques}, and the well-known tangential sphere bound \cite{poltyrev1994bounds}.
However, most of the bounds in the literature are designed based on the assumption of Gaussian noises and are unsuitable for evaluating the performance of impulsive noise channels.
As a special case of the GGD, it shows that the Laplace distribution can accurately model the noises in various systems \cite{beaulieu2008p, beaulieu2008soft, desai2007robust}. Due to its simplicity within the GGD family, the performance analysis of systems assuming Laplace noise is worth studying.
We derive the PEP conditioned on the $L_1$-norm of the noise vector in a closed-form expression. 
It helps to derive both the asymptotic IS gain of the proposed IS estimator and the sphere bound for the additive white Laplace noise (AWLN) channel, which is tighter than the union bound provided in \cite{marks1978detection, shao2012investigation}.

Our main contributions are summarized as follows:
\begin{itemize}
	\item[1.] We present a general framework to designing importance sampling (IS) estimators for memoryless continuous channels. 
	The proposal IS distribution is a generalized spherical distribution with level sets determined by a mapping function that maps the observation space to the real line. 
	We keep the conditional distribution within each level set constant and adjust the densities among these level sets.
	Given an arbitrary mapping function, we derive the optimal IS distribution that minimizes the variance of the IS estimator.
	We consider the mapping functions of summation format and propose a sample generation method for the corresponding proposal IS distribution.
	\item[2.] For the additive white generalized Gaussian noise channel, we choose the $L_p$-norm as the mapping function and derive the corresponding optimal IS distribution, which is an $L_p$-norm spherical ($L_p$-spherical) distribution. 
	For any $L_p$-spherical distributions, we present a sampling scheme, which consists of a transformation method that can generate samples exactly following the multivariate uniform distribution in the $L_p$-sphere.
	\item[3.] For the additive white Laplace noise (AWLN) channel, we derive the pairwise error probability (PEP) conditioned on the $L_1$-norm of the noise vector in a closed-form expression. Based on the conditional PEPs, the sphere bound on the error probability of maximum-likelihood decoding of a binary linear block code is derived. The radius of the $L_1$-sphere is optimized to tighten the bound.
	\item[4.] Finally, we analyze the efficiency of the proposed IS estimator and derive the asymptotic IS gain in a multiple integral form as SNR tends to infinity, where the dimension of the integral is the minimum distance of the code. Specifically, for the AWLN and AWGN channels, the gains can be simplified to a one-dimensional integral based on the derived closed-form conditional PEPs, which significantly reduce the complexity of the numerical calculation.
\end{itemize}

The rest of the paper is organized as follows. In Section \ref{sec:preliminary}, some preliminaries of the importance sampling and the bounding technique are provided. The general framework to designing IS estimators for memoryless continuous channels is presented in Section \ref{sec:optIS}. The proposed IS estimator for the AWGGN channel is described in Section \ref{sec:propIS}. The sphere bound for the AWLN channel is derived in Section \ref{sec:sb}. In Section \ref{sec:ISgain}, the asymptotic IS gain of the proposed IS estimator is presented. Simulation results are shown in Section \ref{sec:results}. Finally, several concluding remarks are given in Section \ref{sec:conclusion}. 

\section{Preliminaries}
\label{sec:preliminary}
Consider an $(n,k)$ linear block code transmitted through a memoryless continuous channel. Denote the channel input by $\mathbf{x} = [x_1,\dots, x_n]^T$ and the channel output by $\mathbf{Y} = [Y_1, \dots, Y_n]^T$, $\mathbf{Y} \in \mathcal{Y}$, where $\mathcal{Y}$ is denoted as the $n$-dimensional observation space. In this paper, we use bold lower case letters, e.g., $\mathbf{x}$ to refer to deterministic samples, and bold upper case letters, e.g., $\mathbf{Y}$ to refer to random vectors. 

Denote $\mathcal{E}$ as the error region of the transmitted signal vector $\mathbf{x}$ and $I(\mathbf{y})$ as an indicator function, which equals 1 if the channel output $\mathbf{y}$ falls inside the error region, and 0 otherwise. The word error rate (WER) $P_e$ can be written as
\begin{equation}
	P_e = \text{Pr}\left(\mathbf{Y} \in \mathcal{E} \right) = \int_{\mathcal{Y}} I(\mathbf{y}) f(\mathbf{y}) d \mathbf{y} ,
	\label{eq:WER}
\end{equation}
where $f(\mathbf{y}) = \prod_{i=1}^n f(y_i)$ and $f(y_i)$ represents the probability density function (p.d.f.) of $Y_i$.

\subsection{Importance Sampling}
The standard MC method estimates the WER in (\ref{eq:WER}), denoted by $\hat{P}_e^{\text{MC}}$, is unbiased and can be written as
\begin{equation}
	\hat{P}_e^{\text{MC}} = \frac{1}{N} \sum_{i=1}^{N} I(\mathbf{y}_i), \quad \mathbf{y}_i \sim f(\mathbf{y}),
	\label{eq:MC_est}
\end{equation}
where $\mathbf{y}_i$ is the $i$-th sample and $N$ is the total number of samples. 

The variance of the estimator is commonly used as the metric to describe its reliability in approaching the groundtrue value of $P_e$ as $N$ increases. It is well-known that the variance of the MC estimator is
\begin{equation}
	\text{Var} \left[ \hat{P}_e^{\text{MC}} \right] = \frac{P_e \left(1-P_e \right)}{N}.
	\label{eq:MC_var}
\end{equation}

Importance sampling is a variance-reduction technique that involves generating samples from a biased distribution $f^*(\mathbf{z})$, called the IS distribution, instead of $f(\mathbf{z})$. Its basic concept is to reformulate the WER in (\ref{eq:WER}) as 
\begin{equation}
	P_e = \int_{\mathcal{Y}} I(\mathbf{y}) \frac{f(\mathbf{y})}{f^*(\mathbf{y})} f^*(\mathbf{y}) d \mathbf{y},
	\label{eq:WER_IS}
\end{equation}
and estimates the weighted average
\begin{equation}
	\hat{P}_e^{\text{IS}} = \frac{1}{N} \sum_{i=1}^{N} I(\mathbf{y}_i) \frac{f(\mathbf{y}_i)}{f^*(\mathbf{y}_i)}, \quad \mathbf{y}_i \sim f^*(\mathbf{y}).
	\label{eq:IS_est}
\end{equation}  

As the expectation of $\hat{P}_e^{\text{IS}}$ is the groundtrue value $P_e$, the IS estimator is also unbiased. Its variance is derived in \cite{xia2003importance} as
\begin{equation}
	\text{Var} \left[ \hat{P}_e^{\text{IS}} \right] = \frac{1}{N} \left(\int_{\mathcal{Y}} I(\mathbf{y}) \frac{f^2(\mathbf{y})}{f^*(\mathbf{y})} d \mathbf{y} - P_e^2 \right).
	\label{eq:IS_var}
\end{equation}

Relative error of the estimator is commonly used as the stopping criterion of the simulation, which is defined as \cite{rubinstein2016simulation}
\begin{equation}
	\kappa \triangleq \frac{\sqrt{\text{Var} \left[ \hat{P}^{\text{IS}}_e \right]}}{P_e}.
	\label{eq:SysM_relative_error}
\end{equation}
Specifically, if the IS distribution is the same as the original sampling distribution, the IS estimator is degraded to the MC estimator. When the error probability is small (i.e. $P_e \ll 1$), the relative error in the MC simulation becomes
\begin{equation}
	\kappa = \sqrt{\frac{1-P_e}{NP_e}}.
\end{equation}
The number of samples needed to achieve a given $\kappa$ can be approximated by
\begin{equation}
	N \approx \frac{1}{\kappa^2 P_e},
	\label{eq:subIS_MC_N}
\end{equation}
which suggests that $N \approx 100/P_e$ samples are required in order to obtain a reliable estimation result with a relative error of $10 \%$, i.e., $\kappa = 0.1$. 

\subsection{AWGGN Channel}
When the AWGGN channel is considered, the channel model can be described as
\begin{equation}
	\mathbf{Y} = \mathbf{x} + \mathbf{Z},
\end{equation}
where $\mathbf{Z} = [Z_1, Z_2, \dots, Z_n]^T$ and $Z_i$'s are independent and identically distributed (i.i.d.) random variables following zero-mean generalized Gaussian distribution (GGD). The p.d.f. of $Z_i$, with variance $\sigma^2$, can be written as
\begin{equation}
	f(z) = \frac{p}{2 \alpha \Gamma(\frac{1}{p})} e^{-\frac{|z|^p}{\alpha^p}}, \quad z \in \mathbb{R}, p>0, \alpha>0,
\end{equation}
where $\mathbb{R}$ is the real line, $\alpha = \sigma \sqrt{\frac{\Gamma(1/p)}{\Gamma(3/p)}}$,
and 
\begin{equation*}
	\Gamma(x) = \int_{0}^{\infty} t^{x-1}e^{-x} dx
\end{equation*} 
is the Gamma function. The parameters $p$ and $\alpha$ control the shape and scale of the distribution, respectively. One remark is that the GGD reduces mathematically to the Laplace and Gaussian distribution when $p=1$ and $p=2$, respectively. 

\subsection{Gallager's First Bounding Technique (GFBT)}
Error probability bounds are widely used to evaluate the MLD performance of a binary linear block code. Many well-known bounds are based on the so-called Gallager's first bounding technique \cite{divsalar1999simple}. The basic concept is to introduce a Gallager region $\mathcal{G}$ around the transmitted signal vector so that the classical union bound can be tightened by avoiding over-estimating the error probability outside the Gallager region.
\begin{align}
	P_e & = \text{Pr}(\mathbf{Y} \in \mathcal{E} \cap \mathcal{G}) + \text{Pr}(\mathbf{Y} \in \mathcal{E}, \mathbf{Y} \notin \mathcal{G})\nonumber \\
	& \le \text{Pr}(\mathbf{Y} \in \mathcal{E} \cap \mathcal{G}) + \text{Pr}(\mathbf{Y} \notin \mathcal{G}).
	\label{eq:GFBT}
\end{align}
The first term on the right-hand side in (\ref{eq:GFBT}) is upper bounded by the union bound. The second term usually dominates the error probability for low SNR.

\section{Proposed IS Estimator Design Framework}
\label{sec:optIS}

In this section, we propose a general framework to designing IS estimators for memoryless continuous channels. 
Based on an arbitrary mapping function that maps the observation space to the real line, we simplify the optimal IS distribution searching problem from $n$-dimensional to one-dimensional. A corresponding IS estimator is therefore proposed.
The choice of the mapping function and the sample generation method based on the derived optimal IS distribution are investigated.

\subsection{Optimal IS Distribution}
As a variance-reduction technique, the efficiency of the IS method is determined by the proposal IS distribution. 
The global optimal IS distribution on the $n$-dimensional observation space $\mathcal{Y}$, which makes the variance of the IS estimator in (\ref{eq:IS_var}) become 0, is already known and provided in \cite{smith1997quick} as
\begin{equation}
	f^{\star}(\mathbf{y}) = \frac{I(\mathbf{y})f(\mathbf{y})}{P_e}.
\end{equation}
However, it requires $P_e$ that needs to be estimated itself and is thus infeasible.
Directly searching for a proper IS distribution for $n$-dimensional random vector $\mathbf{Y}$ that can accelerate the simulation is extremely complex. We simplify the optimal IS distribution searching problem to one-dimensional based on a chosen mapping function.

Consider a mapping function $w: \mathcal{Y} \mapsto \mathcal{R}$, where $\mathcal{R} \subseteq \mathbb{R}$ is a subset of the real line. 
Meanwhile, a random variable can be defined as $R \triangleq w(\mathbf{Y})$ with $\mathcal{R}$ as its sample space. The p.d.f. of $R$ can be expressed as 
\begin{equation}
	g(r) = \int_{\mathcal{Y}_r} f(\mathbf{y}) d\mathbf{y},
\end{equation}
where $\mathcal{Y}_r  = \{\mathbf{y}\in\mathcal{Y}: w(\mathbf{y}) = r\}$ is the level set of $w(\mathbf{y})$ given constant value $r$. Apparently, the sets $\mathcal{Y}_r$'s for different $r$ are mutually disjoint and the union of all $\mathcal{Y}_r$'s is the whole observation space $\mathcal{Y}$.

The distribution of $\mathbf{Y}$ conditioned on the level set $\mathcal{Y}_r$ for the original distribution is 
\begin{equation*}
	f(\mathbf{y}| R= r)  = \frac{I_r(\mathbf{y}) f(\mathbf{y})}{g(r)},
\end{equation*}
where $I_r(\mathbf{y})$ is the indicator function that returns 1 if $w(\mathbf{y}) = r$ and 0 otherwise.
We construct our proposal IS distribution by keeping the conditional distribution $f(\mathbf{y}| R= r)$ constant and adjusting the density function $g(r)$ as
\begin{equation}
	f^*(\mathbf{y}) = \int_{\mathcal{R}} f(\mathbf{y}| R= r) g^*(r) dr = \frac{f(\mathbf{y})}{g(w(\mathbf{y}))} \cdot g^*(w(\mathbf{y})),
	\label{eq:IS_dist}
\end{equation}
where $g^*(r)$ can be regarded as the IS distribution for $R$. 

Therefore, the proposed IS estimator can be written as
\begin{align}
	\hat{P}_e^{\text{IS}} = \frac{1}{N} \sum_{i=1}^N I(\mathbf{y}_i) \frac{f(\mathbf{y}_i)}{f^*(\mathbf{y}_i)}=  \frac{1}{N} \sum_{i=1}^N I(\mathbf{y}_i) \frac{g(w(\mathbf{y}_i))}{g^*(w(\mathbf{y}_i))}, \quad \mathbf{y}_i \sim f^*(\mathbf{y}).
	\label{eq:ISest_cc}
\end{align}

By instituting $f^*(\mathbf{z})$ in (\ref{eq:IS_dist}) into (\ref{eq:IS_var}), the variance of the proposed IS estimator can be derived as
\begin{align}
	\text{Var} \left[\hat{P}_e^{\text{IS}}\right] &= \frac{1}{N} \left(\int_{\mathcal{Y}} I(\mathbf{y}) \frac{f^2(\mathbf{y})}{f^*(\mathbf{y})} d\mathbf{y} - P_e^2 \right) \nonumber \\
	& = \frac{1}{N} \left( \int_{\mathcal{Y}} I(\mathbf{y}) \frac{g(w(\mathbf{y}))}{g^*(w(\mathbf{y}))} f(\mathbf{y}) d\mathbf{y} -P_e^2 \right) \nonumber \\
	&=  \frac{1}{N} \left( \int_{\mathcal{R}} \frac{g(r)}{g^*(r)} \int_{\mathcal{Y}_r} I(\mathbf{y})  f(\mathbf{y}) d\mathbf{y} dr -P_e^2 \right)  \nonumber \\
	&= \frac{1}{N} \left( \int_{\mathcal{R}} \theta(r) \frac{g^2(r)}{g^*(r)}dr -P_e^2\right),
	\label{eq:IS_var_r}
\end{align}
where 
\begin{equation}
	\theta(r) \triangleq \Pr\left( I(\mathbf{Y})=1| w(\mathbf{Y}) = r\right) = \frac{1}{g(r)} \int_{\mathcal{Y}_r} I(\mathbf{y})  f(\mathbf{y}) d\mathbf{y},
	\label{eq:error_ratio}
\end{equation}
named as the error ratio, is defined as the error probability conditioned on $\mathcal{Y}_r$.

For a given mapping function $w(\mathbf{y})$, the following theorem provides the optimal IS distribution $g^{\star}(r)$ over all of the possible distributions on the sample space $\mathcal{R}$ that minimizes the variance of the estimator. 
\begin{theorem}
	\label{th:optIS}
	Consider a memoryless continuous channel with a received random vector $\mathbf{Y}$. Given an arbitrary function $w: \mathcal{Y} \mapsto \mathcal{R}$ and denote $R= w(\mathbf{Y})$ as a random variable with p.d.f. $g(r)$. Then, the optimal IS distribution for $R$ that minimizes the variance of the IS estimator in (\ref{eq:IS_var_r}) is given by
	\begin{equation}
		g^{\star}(r) = \frac{\sqrt{\theta(r)} g(r)}{\int_{\mathcal{R}} \sqrt{\theta(r)}g(r) dr}, \quad r \in \mathcal{R}. 
		\label{eq:opt_IS_r}
	\end{equation}
\end{theorem}
\begin{proof}
	Define the random variable $X \triangleq \frac{\sqrt{\theta(R)}g(R)}{g^*(R)}$ and the convex function $\phi(x) = x^2$. According to Jensen's inequality in {\cite[Theorem 1.5.1]{durrett2019probability}}, we can get
	\begin{align*}
		E\left[\phi(X)\right] &\ge \phi(E\left[ X\right]) \\
		\Rightarrow \quad E \left[ \left( \frac{\sqrt{\theta(R)}g(R)}{g^*(R)}\right)^2\right] &\ge E^2 \left[ \frac{\sqrt{\theta(R)}g(R)}{g^*(R)}\right] \\
		\Rightarrow \quad \int_{\mathcal{R}} \left( \frac{\sqrt{\theta(r)}g(r)}{g^*(r)}\right)^2 \cdot g^*(r) dr &\ge \left( \int_{\mathcal{R}} \frac{\sqrt{\theta(r)}g(r)}{g^*(r)} \cdot g^*(r) dr\right)^2\\
		\Rightarrow \quad \int_{\mathcal{R}} \frac{\theta(r) g^2(r)}{g^*(r)} dr &\ge \left( \int_{\mathcal{R}} \sqrt{\theta(r)} g(r) dr\right)^2,
	\end{align*}
	where the left-hand side is the first term of the variance of the proposed IS estimator in (\ref{eq:IS_var_r}), and the right-hand side is a constant. Thus, we can get the following inequality for the derived variance
	\begin{equation}
		\text{Var} \left[\hat{P}_e^{\text{IS}}\right] \ge \frac{1}{N}\left( \left(\int_{\mathcal{R}} \sqrt{\theta(r) }g(r) dr \right)^2 -P_e^2 \right).
		\label{eq:opt_var_IS}
	\end{equation}
	
	According to Theorem 1.4.7 in \cite{durrett2019probability}, the equality holds if and only if $\phi(x)$ is affine or $X$ is a constant almost surely.
	Since $\phi(x) = x^2$ is strictly convex for $x >0$, it follows that $X = \frac{\sqrt{\theta(R)}g(R)}{g^*(R)}$ is a constant almost surely and
	\begin{equation*}
		\left(\frac{\sqrt{\theta(r)}g(r)}{g^*(r)} \right)^2 \int_{\mathcal{R}}g^*(r) dr = \left( \int_{\mathcal{R}} \sqrt{\theta(r)} g(r) dr\right)^2.
	\end{equation*} 
	Therefore, the optimal IS distribution of $R$ is 
	\begin{equation*}
		\begin{aligned}
			g^{\star}(r) = \frac{\sqrt{\theta(r)} g(r)}{\int_{\mathcal{R}} \sqrt{\theta(r)}g(r) dr}.
		\end{aligned}
	\end{equation*}
\end{proof}

Unfortunately, the optimal IS distribution contains the error ratio $\theta(r)$, which itself needs to be estimated.
The adaptive IS can be adopted to solve the problem. By regarding the error ratios as the parameters of the IS distribution, we gradually approach the optimal IS distribution by updating these parameters iteratively during the simulation.

\subsection{Mapping Function and the Sample Generation Method}
The minimized variance shown in the right-hand side of the inequality (\ref{eq:opt_var_IS}) is determined by the choice of the mapping function $w(\mathbf{y})$.
When we choose $w(\mathbf{y})$, we need to take both the variance minimization and the efficiency of the sample generation into consideration. 
We can see that the variance becomes zero if $\sqrt{\theta(r)} = \theta(r), \forall r \in \mathcal{R}$, which means the error ratio becomes an indicator function. This leads to the mapping function whose level sets $\mathcal{Y}_r$'s are scaled decision regions. 

\begin{figure}[!h]
	\centering
	\includegraphics[width=3in]{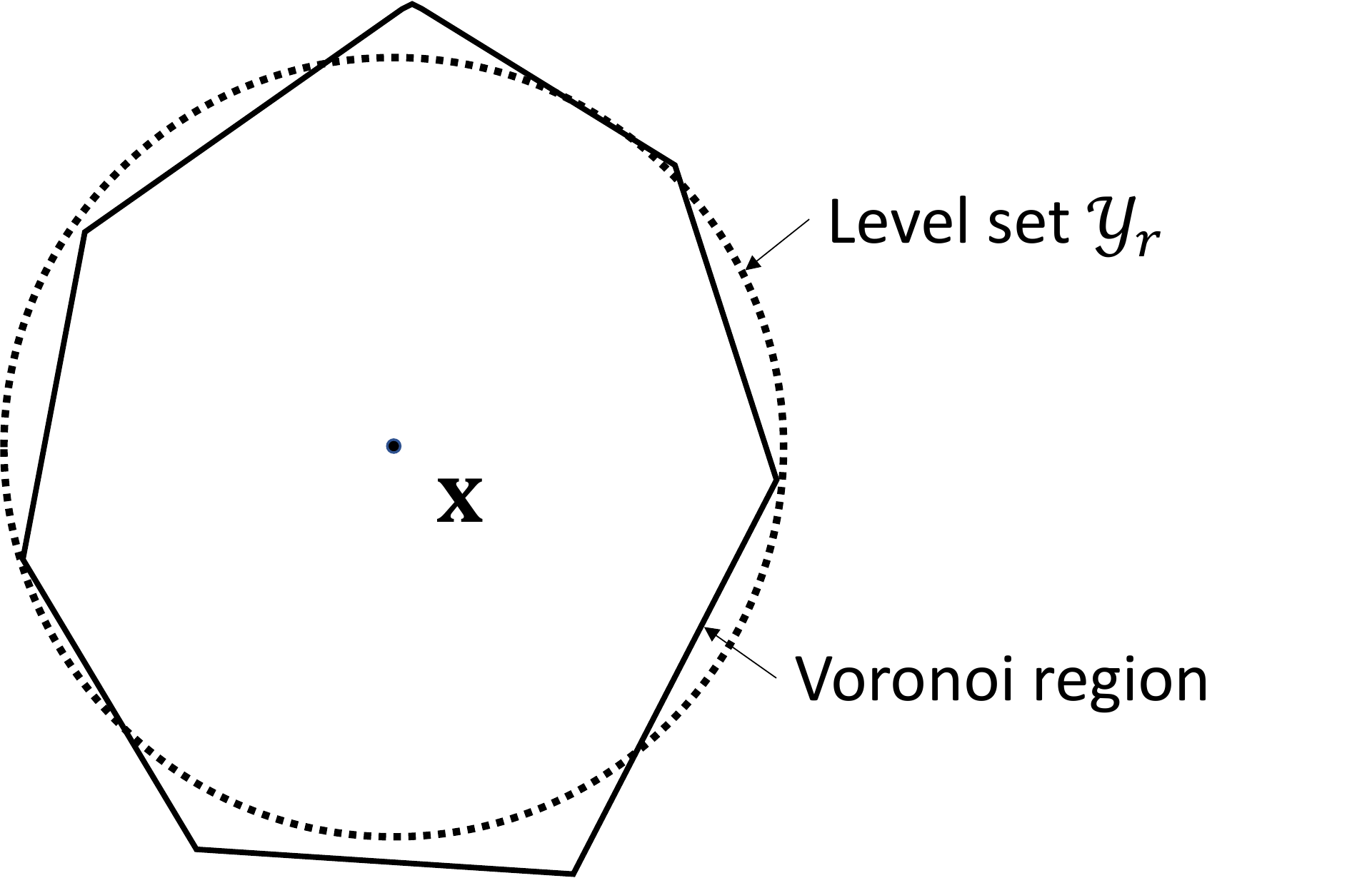} 
	\caption{Geometric interpretation of the level set of the mapping function and the Voronoi region of $\mathbf{x}$.}
	\label{fig:Voronoi}
\end{figure}

In most cases, the knowledge about the decision region of the code is absent. According to the maximum-uncertainty principle, a function that makes each $\mathcal{Y}_r$ an equal-density contour can be a candidate for this extreme case. It can be illustrated as having the equivalent probability density for all directions in the geometric interpretation. Therefore, such a mapping function must be some variant of the joint p.d.f., e.g., the negative log-likelihood
\begin{equation}
	w(\mathbf{y}) = -\sum_{i=1}^{n} \log f(y_i).
	\label{eq:neg_LL}
\end{equation} 
In general, the closed-form expression of the p.d.f. $g(r)$ is intractable. However, since $g(r)$ is one-dimensional, it is worth numerically estimating and tabulating $g(r)$ offline compared to the time saving from the IS. 

From the structure of the IS distribution, we can roughly divide the sample generation process into two phases. During the first phase, a scalar $r'$ is randomly generated from the p.d.f. $g^*(r)$, which is numerically straightforward even if its closed-form expression is unavailable. The challenge is in the second phase about drawing samples from the level set $\mathcal{Y}_{r'}$. Suppose the mapping function takes the form 
\begin{equation}
	w(\mathbf{y}) = \sum_{i=1}^{n} \nu(y_i), \quad \text{ for } \mathbf{y} \in \mathcal{Y},
\end{equation}
where $\nu(\cdot)$ is an arbitrary function so that the following transformation 
\begin{align}
	\left\{ \begin{matrix*}[l]
		u_i = \frac{\nu(y_i)}{w(\mathbf{y})}, \quad \text{ for } i = 1,2,\cdots,n-1, \\
		r = w(\mathbf{y}),
	\end{matrix*}
	\right.
	\label{eq:transform}
\end{align}
is 1-to-1 correspondence in the domain of $(u_1,\cdots,u_{n-1},r)$.

Denote $\rho = \nu^{-1}$ as the inverse function of $\nu$. Then,
\begin{align*}
	\left\{ \begin{matrix*}[l]
		y_i = \rho(ru_i), \quad \text{ for } i = 1,2,\cdots,n-1, \\
		y_n = \rho \left(r - r\sum_{i=1}^{n-1} u_i\right).
	\end{matrix*}
	\right.
\end{align*}
Hence, the Jacobian of the transformation is
\begin{align*}
	J(u_1,\cdots,u_{n-1},r) = \left| \begin{matrix}
		r\rho'(ru_1) & \cdots & 0 & -r \rho'(ru_n) \\
		\vdots & \ddots & \vdots & \vdots \\
		0 & \cdots & r\rho'(ru_{n-1}) & -r\rho'(ru_n) \\
		u_1\rho'(ru_1) & \cdots & u_{n-1}\rho'(ru_{n-1}) & u_n \rho'(ru_n)
	\end{matrix} \right| = r^{n-1} \prod_{i=1}^{n} \rho'(r u_i),
\end{align*}
where $u_n = 1- \sum_{i=1}^{n-1} u_i$ for simplicity.

Then the joint p.d.f. of $U_1, \cdots, U_{n-1}, R$ is 
\begin{align}
	f(u_1,\cdots,u_{n-1},r) & = J(u_1,\cdots,u_{n-1},r) f\left(\rho\left(r-r\sum_{i=1}^{n-1} u_i \right)\right) \prod_{i=1}^{n-1} f\left(\rho(ru_i)\right) \nonumber \\
	& = r^{n-1} \prod_{i=1}^n  \rho'(r u_i) f\left(\rho(ru_i)\right).
\end{align}
Therefore, the marginal distribution of $R$ is 
\begin{equation*}
	g(r) = \int_{\mathcal{Y}_r} f(\mathbf{y}) d \mathbf{y} = \underbrace{\int \cdots \int}_{n-1} f(u_1,\cdots, u_{n-1},r) du_1 \cdots d u_{n-1}.
\end{equation*}
and the distribution of $U_1,\cdots,U_{n-1}$ conditioned on the radius $r$ is
\begin{align}
	f(u_1,\cdots, u_{n-1}|r) = \frac{1}{g(r)} f(u_1,\cdots,u_{n-1},r) = \frac{r^{n-1}}{g(r)} \prod_{i=1}^n  \rho'(r u_i) f\left(\rho(ru_i)\right).
	\label{eq:pdf_r}
\end{align}
We can use Markov chain Monte Carlo methods like the Metropolis algorithm to generate samples from the above conditional distribution.

Specifically, if the negative likelihood (\ref{eq:neg_LL}) is chosen as the mapping function (i.e., $\nu(y) = -\log f(y)$), then $\rho (x) = f^{-1}(e^{-x})$, and the conditional distribution becomes
\begin{align*}
	f(u_1,\cdots,u_{n-1}|r) &= -\frac{r^{n-1}}{g(r)} \prod_{i=1}^{n} e^{-ru_i} \left(f^{-1}\right)' (e^{-ru_i}) f\left(f^{-1}(e^{-ru_i})\right)\\
	&= -\frac{r^{n-1}e^{-2r}}{g(r)} \prod_{i=1}^{n} \left(f^{-1}\right)'(e^{-ru_i})\\
	& = -\frac{r^{n-1}e^{-2r}}{g(r)} \prod_{i=1}^{n} \frac{1}{f'\left(f^{-1}(e^{-ru_i})\right)}.
\end{align*}

For the function $\nu$ that is not invertible, we define $\rho$ as a multivalued inverse of $\nu$ that consists of invertible branches defined on their restricted domains, respectively.
The transformation is considered for the domain of each branch separately, so that a 1-to-1 correspondence is guaranteed.
A sample generation method can be proposed as follows.
\begin{itemize}
	\item[1. ] A sample $y'$ is generated from the original distribution $f(y)$. Corresponding to the domain that $y'$ falls in, a branch is determined. 
	\item[2. ] A radius $r'$ is generated from the p.d.f. $g^*(r)$ conditioned on the determined branch.
	\item[3. ] A vector $[u_1, u_2, \cdots, u_n]^T$ is sampled following $f(u_1,\cdots,u_{n-1}|r')$ in (\ref{eq:pdf_r}) using Metropolis algorithm.
	\item[4. ] A sample $\mathbf{y} = [y_1,\cdots,y_n]^T$ that following the IS distribution $f^*(\mathbf{y})$ is obtained by setting $y_i = \rho(r'u_i)$ for $i=1,2,\cdots,n$.
\end{itemize}

One remark is that if $U_1,\cdots, U_{n-1}$ are independent from $R$ (i.e., $f(u_1,\cdots, u_{n-1},r) = g(r)f(u_1,\cdots, u_{n-1})$), the above step 3 can be replaced by a more efficient and straightforward sampling method that does not involve calculation of $f(u_1,\cdots,u_{n-1}|r')$ as follows.
\begin{itemize}
	\item[1. ] Draw $n$ independent samples $y'_1,y'_2,\cdots, y'_n$ from $f(y)$ conditioned on the determined branch. 
	\item[2. ] Set $u_i = \frac{\nu(y'_i)}{\sum_{j=1}^n \nu(y'_j)} $, for $i = 1,\cdots, n$, so that a random vector $[u_1, \cdots, u_n]^T$ following $f(u_1,\cdots, u_{n-1})$ is sampled.
\end{itemize}

\section{Proposed IS Estimator for the AWGGN Channel}
\label{sec:propIS}
In this section, we apply the proposed framework to the AWGGN channel. By choosing the $L_p$-norm as the mapping function, we propose an $L_p$-norm-based minimum-variance IS estimator. A sample generation process that can uniformly draw samples from the $L_p$-sphere for the IS simulation is presented, and the details of the IS algorithm are shown.

Since the code is linear and the channel is symmetric, without loss of generality, assume the signal vector $\mathbf{x}_0$ is transmitted. As the channel output is determined by the noise, we consider the noise random vector $\mathbf{Z}$ instead of $\mathbf{Y}$ for simplicity. Without any side information about the decision region of the code, a variant of the joint p.d.f. $f(\mathbf{z})$ is chosen as the mapping function
\begin{equation*}
	w(\mathbf{z}) = \alpha \left(n \log \frac{p}{2 \alpha\Gamma(\frac{1}{p})} - \log f(\mathbf{z}) \right)^{\frac{1}{p}} = \|\mathbf{z}\|_p,
\end{equation*}
where $\|\cdot\|_p$ stands for the $L_p$-norm, and 
\begin{equation*}
	f(\mathbf{z}) = \left(\frac{p}{2\alpha \Gamma(\frac{1}{o})}\right)^n e^{-\frac{1}{\alpha^p} \|\mathbf{z}\|_p^p}.
\end{equation*}
Define $R \triangleq w(\mathbf{Z})$ as a random variable. The resultant level set $\mathcal{Y}_r$ is an $L_p$-sphere centered at $\mathbf{x}_0$ with radius $r$, which is an equal-density contour in $\mathcal{Y}$. Meanwhile, the error ratio $\theta(r)$ defined in (\ref{eq:error_ratio}) for the $L_p$-sphere is independent from the SNR. This leads to another advantage that when we estimate the WER w.r.t. the SNR, $\theta(r)$ estimated in low SNR can be tabulated and used to accelerate the simulation for high SNR.

The p.d.f. $g(r)$ can be derived based on the following lemma.
\begin{lemma}
	\label{lemma:gamma}
	Let $\mathbf{Z} = \left[Z_1, Z_2, \cdots, Z_n\right]^T$, where $Z_i$'s are i.i.d. random variables following generalized Gaussian distribution.
	Let $X = \|\mathbf{Z}\|_p^p$. Then $X \sim \Gamma\left(\frac{n}{p}, \alpha^p\right)$ follows Gamma distribution and its cumulative density function (c.d.f.) and p.d.f. can be expressed as
	\begin{align}
		&F_X(x;n) = \gamma\left(\frac{x}{\alpha^p}, \frac{n}{p}\right), \qquad x\ge 0 \\
		&f_{X}(x;n) = \frac{1}{\Gamma(\frac{n}{p})\alpha^n} x^{\frac{n}{p}-1} e^{-\frac{x}{\alpha^p}}, \qquad x \ge 0,
	\end{align}
	respectively, where 
	\begin{equation*}
		\gamma(x,s) = \frac{1}{\Gamma(s)} \int_0^x t^{s-1} e^{-t} dt
	\end{equation*}
	is the regularized lower incomplete Gamma function.
\end{lemma}
\begin{proof}
	Let $X = \sum_{i=1}^n |Z_i|^p$. Its characteristic function can be derived as
	\begin{equation}
		\Phi(t) = E\left[ e^{jtX} \right] = \prod_{i=1}^{n} E\left[ e^{jt|Z_i|^p}\right] = \left(1-jt \alpha^p\right)^{-\frac{n}{p}},
		\label{eq:charac_func}
	\end{equation}
	where
	\begin{align*}
		E\left[ e^{jt|Z_i|^p}\right] &= \int_{-\infty}^{\infty} e^{jt|z|^p} \frac{p}{2\alpha \Gamma(\frac{1}{p})} e^{-\frac{|z|^p}{\alpha^p}}dz =  \frac{p}{\alpha \Gamma(\frac{1}{p})}\int_{0}^{\infty} e^{jtz^p} e^{-\frac{z^p}{\alpha^p}}dz \\
		&= \frac{p}{\alpha \Gamma(\frac{1}{p})} \int_{0}^{\infty} e^{-x} d\left(\frac{x}{\frac{1}{\alpha^p}-jt}\right)^{\frac{1}{p}} = \frac{1}{\alpha \Gamma(\frac{1}{p})} \left(\frac{\alpha^p}{1-jt\alpha^p}\right)^{\frac{1}{p}} \int_{0}^{\infty} e^{-x}x^{\frac{1}{p}-1} dx \\
		&= \left(1-jt \alpha^p\right)^{-\frac{1}{p}}
	\end{align*}
	The characteristic function (\ref{eq:charac_func}) indicates that $X$ follows Gamma distribution with shape parameter $\frac{n}{p}$ and scale parameter $\alpha^p$.
\end{proof}

As $R = X^{\frac{1}{p}}$, its c.d.f. is
\begin{align*}
	F_R(r) = \text{Pr}\left(R < r\right) = \text{Pr}\left(X < r^p\right) = F_{X}(r^p;n), \quad r \ge 0.
\end{align*}
Hence, the p.d.f. of $R$ can be derived as
\begin{align}
	g(r) = \frac{d}{dr} F_R(r) = f_{X}(r^p;n)\frac{d}{dr}r^p = \frac{p}{\Gamma(\frac{n}{p})\alpha^n}r^{n-1} e^{-\frac{r^p}{\alpha^p}}, \quad r \ge 0.
	\label{eq:pdf_R}
\end{align}

Following the framework in Section \ref{sec:optIS}, the proposed IS estimator for the AWGGN channel can be written as
\begin{equation}
	\hat{P}_e^{\text{IS}} = \frac{1}{N} \sum_{i=1}^N I(\mathbf{z}_i) \frac{g(\|\mathbf{z}_i\|_p)}{g^{\star}(\|\mathbf{z}_i\|_p)}, \quad \mathbf{z}_i \sim f^*(\mathbf{z}) ,
	\label{eq:ISest_GGN}
\end{equation}
where $I(\mathbf{z})$ returns 1 if $\mathbf{z}$ falls outside the decision region of $\mathbf{x}_0$ and 0 otherwise, $f^*(\mathbf{z}) = \frac{g^{\star}(\|\mathbf{z}\|_p)}{g(\|\mathbf{z}\|_p)} f(\mathbf{z}) $, and the optimal IS distribution $g^{\star}(r)$ is provided in Theorem \ref{th:optIS}.

To generate samples from the proposal IS distribution $f^*(\mathbf{z})$, we need to further derive the conditional distribution on $\mathcal{Y}_r$.
Denote $\nu(z) = |z|^p$ so that $w(\mathbf{z}) = \left(\sum_{i=1}^{n} \nu(z_i) \right)^{\frac{1}{p}}$. The multivalued inverse of $\nu$ is $\rho(u) = \pm u^{\frac{1}{p}}$. 
To get an 1-to-1 transformation, we partition $\mathcal{Y}$ into $2^n$ domains based on the sign of each element of $\mathbf{z}$, denoted as $\text{sgn}(z_i)$, for $i = 1,\cdots, n$. We consider the following transformation in each domain
\begin{align*}
	&\left\{ \begin{matrix*}[l]
		u_i = \frac{|z_i|^p}{\|\mathbf{z}\|_p^p}, \quad \text{for } i = 1,\cdots,n-1, \\
		r = \|\mathbf{z}\|_p,
	\end{matrix*}
	\right. \\
	\Rightarrow \quad 
	&\left\{ \begin{matrix*}[l]
		z_i = \text{sgn}(z_i) r u_i^{\frac{1}{p}}, \quad \text{for } i = 1,\cdots,n-1, \\
		z_{n} = \text{sgn}(z_n) r \left(1 - \sum_{i=1}^{n-1} u_i \right)^{\frac{1}{p}}.
	\end{matrix*}
	\right.
\end{align*}
The resultant Jacobian has a unified expression that does not depend on the domain
\begin{equation*}
	J(u_1, \cdots, u_{n-1}, r) = (r^p)^{n-1} \prod_{i=1}^{n} \rho'(r^p u_i) \cdot \frac{d}{dr}(r^p) = \frac{r^{n-1}}{p^{n-1}} \left(1- \sum_{i=1}^{n-1} u_i\right)^{\frac{1}{p}-1} \prod_{i=1}^{n-1} u_i^{\frac{1}{p}-1}.
\end{equation*} 

Therefore, the conditional distribution on $\mathcal{Y}_r$ is 
\begin{align}
	f(u_1,\cdots, u_{n-1}|r) = \frac{\Gamma(\frac{n}{p})}{2^n \Gamma^n(\frac{1}{p})} \left(1- \sum_{i=1}^{n-1} u_i\right)^{\frac{1}{p}-1} \prod_{i=1}^{n-1} u_i^{\frac{1}{p}-1},
	\label{eq:pdf_ur}
\end{align}
which is independent of $r$. This means we can generate a sample on unit $L_p$-sphere first and map it to $\mathcal{Y}_r$ by multiplying $r$ to generate a sample from $\mathcal{Y}_r$.

As we mentioned before, the optimal $g^{\star}(r)$ contains $\theta(r)$ that needs to be estimated itself. In practice, we need to estimate $\theta(r)$ and $P_e$ iteratively. However, as $\mathcal{R}$ is the non-negative real line for the chosen mapping function, it is impossible to estimate infinite number of $\theta(r)$'s. Hence, we quantize the space $\mathcal{R}$ and use a histogram to approximate the optimal IS distribution. 

Set a range $[r_{\text{min}},r_{\text{max}}]$ within $\mathcal{R}$, so that the probabilities $\text{Pr}(R<r_{\text{min}})$ and $\text{Pr}(R \ge r_{\text{max}})$ are negligible. We partition the observation space $\mathcal{Y}$ into $m$ shells
\begin{equation*}
	\mathcal{Y}_l = \left\{\mathbf{z} \in \mathcal{Y} : r_{l-1} \le \|\mathbf{z}\|_p < r_l \right\}, \text{ for } l = 1,2,\dots,m,
\end{equation*}
where $r_l = r_{\text{min}} + l\Delta r$, and $\Delta r = r_{l}-r_{l-1} = \frac{1}{m} \left(r_{\text{max}} - r_{\text{min}}\right)$ is the thickness of each shell $\mathcal{Y}_l$. 
Let $P_l = \int_{r_{l-1}}^{r_l} g(r) dr$ and $\theta_l = \text{Pr}\left(I(\mathbf{Z}) = 1 | \mathbf{Z} \in\mathcal{Y}_l \right)$. 
As $\Delta r$ tends to 0, the following approximations hold
\begin{align*}
	& P_l = \int_{r_{l-1}}^{r_l} g(r) dr \simeq g(r_l) \Delta r,\\
	& \theta_l = \frac{\int_{r_{l-1}}^{r_l} \theta(r) g(r)dr}{\int_{r_{l-1}}^{r_l} g(r) dr} \simeq \theta(r_l).
\end{align*} 
Denote $P^*_l$, for $l = 1,2,\dots,m$, as the actual probability mass function (p.m.f.) adopted on the space $\mathcal{R}$ in the IS simulation. 
The optimal p.m.f. according to Theorem \ref{th:optIS} can be derived as
\begin{equation}
	P^*_l = \int_{r_{l-1}}^{r_l} \frac{ \sqrt{\theta(r)}  g(r)}{\int_{r_{\text{min}}}^{r_{\text{max}}} \sqrt{\theta(r)} g(r)dr} dr= \frac{\sqrt{\theta_l} P_l}{\sum_{i=1}^{l} \sqrt{\theta_l} P_l}.
	\label{eq:optIS_pmf}
\end{equation}
Consequently, the implemented IS estimator can be written as
\begin{equation}
	\hat{P}_e^{\text{IS}} = \frac{1}{N} \sum_{i=1}^{N} I(\mathbf{z}_i) \frac{g(\|\mathbf{z}_i\|_p) \Delta r}{P^*_{l(\mathbf{z}_i)}}, \quad \mathbf{z}_i \sim f^*(\mathbf{z}),
	\label{eq:WER_IS_est}
\end{equation}
where $l(\mathbf{z}_i)$ returns the index of the shell $\mathcal{Y}_l$ that $\mathbf{z}_i$ falls in, and $f^*(\mathbf{z}) = \frac{P^*_{l(\mathbf{z})}}{g(\|\mathbf{z}\|_p)\Delta r}f(\mathbf{z}) $. 

In summary, the sample generation from the above IS distribution $f^*(\mathbf{z})$ can be summarized as follows.
\begin{itemize}
	\item[1. ] An index $l'$ is randomly generated according to the p.m.f. $\{P_l^*\}_{l=1}^{m}$. Then, a radius $r'$ is uniformly drawn from the interval $[r_{l'-1},r_{l'})$.
	\item[2. ] Draw $n$ independent samples $z'_1,z'_2,\cdots, z'_n$ from $f(z)$ and set $u_i = \frac{|z'_i|^p}{\sum_{j=1}^{n} |z'_j|^p} $, for $i = 1,2, \cdots, n$.
	\item[3. ] A noise vector $\mathbf{z} = [z_1,\cdots, z_n]^T$ following $f^*(\mathbf{z})$ is drawn by setting $z_i = b_i \cdot r' u_i^{\frac{1}{p}}$, for $i=1,\cdots,n$, where $b_i$ is uniformly drawn from $\{-1,+1\}$.
\end{itemize}

In Algorithm \ref{algo:proposed}, we summarize the above results and present more details of the proposed IS algorithm.
Initially, the sample size counter {\fontfamily{phv}\selectfont{N\_tot}} is set as 0 and the relative error is set as 1. Since no side information about the coding scheme is provided, we assume that no errors can be corrected (i.e., $\hat{\theta}_l^{(0)} = 1$, for $l =1 ,2,\dots,m$). 

For the $T$-th iteration, the IS p.m.f. is firstly updated as
\begin{equation}
	P_l^{*(T)} = \frac{\sqrt{\hat{\theta}^{(T-1)}_{l}}g(r_l)}{\sum_{j=1}^m \sqrt{\hat{\theta}^{(T-1)}_{j}}g(r_{j})}, \quad \text{for } l=1,2,\cdots,m.
	\label{eq:ISpdf_opt_T}
\end{equation}
New samples are generated based on the updated p.m.f. in (\ref{eq:ISpdf_opt_T}).

Next, since the error ratios are SNR-invariant, they are updated based on all the generated samples as
\begin{equation}
	\hat{\theta}^{(T)}_l = \frac{\sum_{i=1}^N I_l(\mathbf{z}_i) I(\mathbf{z}_i)}{\sum_{i=1}^N I_l(\mathbf{z}_i)}, \quad \text{for } l = 1,2,\cdots, m,
	\label{eq:error_ratio_T}
\end{equation}
where $I_l(\mathbf{z})$ is an indicator function which returns 1 if $\mathbf{z} \in \mathcal{Y}_l$  and 0 otherwise. The WER is estimated inside the loop until the relative error (\ref{eq:SysM_relative_error}) meets the reliability requirement or the maximum number of iteration is reached.

In order to avoid the violation due to the insufficient number of samples during the first several iterations, we set a minimum number of samples {\fontfamily{phv}\selectfont{N\_min}} for $P_e$ and $\theta_l$ to start the update. At the same time, a step size {\fontfamily{phv}\selectfont{N\_step}} is defined to control the update frequency.

\begin{algorithm}[!h]
	\label{algo:proposed}
	\SetKwInput{KwData}{Input}
	\SetKwInput{KwResult}{Output}
	\SetAlgoLined
	\caption{$L_p$-Norm-Based IS Algorithm for the AWGGN Channel}
	\KwData{$E_s/N_0$ and relative error {\fontfamily{phv}\selectfont{re}}}
	\KwResult{WER $\hat{P}_e$ and the sample size {\fontfamily{phv}\selectfont{N\_tot}}}
	
	\textbf{Initialization}: {\fontfamily{phv}\selectfont{N\_tot}} := 0, {\fontfamily{phv}\selectfont{WERre}} := 1, initialize $\{P^*_l\}_{l=1}^m$ with (\ref{eq:ISpdf_opt_T}), where all the error ratios are set to be 1\; 
	\While{{\fontfamily{phv}\selectfont{WERre}} $>$ {\fontfamily{phv}\selectfont{re}}}{
		Generate a radius $r$ based on the p.m.f. $\{P^*_l\}_{l=1}^m$ \;
		Draw a noise vector $\mathbf{z}$ uniformly from the $L_p$-sphere with radius $r$\;
		Set $\mathbf{y} = \mathbf{x}_0+ \mathbf{z}$ and pass it through the decoder \; 
		\If{{\fontfamily{phv}\selectfont{N\_tot}} $>$ {\fontfamily{phv}\selectfont{N\_min}}}{
			Compute $\hat{P}_e$ according to (\ref{eq:WER_IS_est})\;
			Compute {\fontfamily{phv}\selectfont{WERre}} according to (\ref{eq:SysM_relative_error})\;
			\If{mod({\fontfamily{phv}\selectfont{N\_tot}}, {\fontfamily{phv}\selectfont{N\_step}}) == 0}{
				Update $\hat{\theta}_l$ with (\ref{eq:error_ratio_T})\;
				Update $\{P^*_l\}_{l=1}^m$ with (\ref{eq:ISpdf_opt_T})\;
			}
		}
		{\fontfamily{phv}\selectfont{N\_tot}} := {\fontfamily{phv}\selectfont{N\_tot}} + 1\;
	}
	\textbf{return} $\hat{P}_e$ and {\fontfamily{phv}\selectfont{N\_tot}}\;
	
\end{algorithm}

For high SNR, we know that errors near the decision region dominate the error performance. Reflected in the proposed IS algorithm, almost all of these errors fall inside shells $\mathcal{Y}_l$ with $r_l$ that is close to the packing radius. Usually, $\theta_l$'s for these dominant shells are significantly small. If there are no errors generated, some of them may become 0 after the update with (\ref{eq:error_ratio_T}). The corresponding $P^*_l$'s will be frozen if these $\theta_l$'s are used in (\ref{eq:ISpdf_opt_T}), and no samples will be further generated from these shells.
As $\theta(r)$ is a monotonically increasing function w.r.t. $r$, when we update $P^*_l$'s, we set these zero $\theta_l$'s as the values of their nearest nonzero $\theta_l$'s with larger $l$, respectively. This ensures the convergence of the IS distribution and the unbiasedness of the estimator.

On the other hand, we can see from (\ref{eq:error_ratio_T}) that the estimation for the error ratios is still based on the MC method. If we can design a mapping function based on some side information about the coded system so that the resultant error ratios for these dominant shells are larger, the efficiency of the IS estimator in terms of the sample size can be further improved.

\section{Sphere Bound for AWLN Channel}
\label{sec:sb}
In this section, we investigate the PEP conditioned on the $L_p$-sphere, which plays an essential role in the efficiency analysis of the proposed IS estimator in Section \ref{sec:ISgain}. On the other hand, the conditional PEP is also of crucial significance in the derivation of sphere bounds on the MLD performance of the linear block code for memoryless continuous channels. Specifically for the AWLN channel, we derive the conditional PEP in a closed-form expression. The corresponding sphere bound is thus derived, where the radius of the $L_1$-sphere is optimized to tighten the bound.

For memoryless continuous channels, given a mapping function $w:\mathcal{Y} \mapsto \mathcal{R}$ and choose $\mathcal{Y}_{r'}$ with parameter $r'$ as the Gallager region, the GFBT-based bound in (\ref{eq:GFBT}) can be written as
\begin{align}
	P_e \le \text{Pr}(\mathbf{Y} \in \mathcal{E} \cap \mathcal{G}) + \text{Pr}(\mathbf{Y} \notin \mathcal{G}) = \int_{-\infty}^{r'} \theta(r) g(r) dr + \int_{r'}^{\infty} g(r) dr.
	\label{eq:gfbt_sb}
\end{align}
By applying the union bound to the first term on the right-hand side in (\ref{eq:gfbt_sb}), the error ratio can be upper bounded by
\begin{equation*}
	\theta(r) \le \sum_{d = 1}^{n} A_d \text{Pr} (\mathbf{c}_0 \rightarrow \mathbf{c}_d | \mathbf{Y} \in \mathcal{Y}_r),
\end{equation*}
where $\text{Pr}(\mathbf{c}_0 \rightarrow \mathbf{c}_d | \mathbf{Y} \in \mathcal{Y}_r)$ is the PEP between the transmitted all-zero codeword $\mathbf{c}_0$ and a weight-$d$ codeword $\mathbf{c}_d$ conditioned on $\mathcal{Y}_r$ and $A_d$ is the number of weight-$d$ codewords. 

By setting the derivative of the right-hand side in (\ref{eq:gfbt_sb}) w.r.t. $r'$ to 0, the optimal $r'$ that tightens the bound can be determined, which is the root of the following equation
\begin{equation*}
	\sum_{d = 1}^{n} A_d \text{Pr} (\mathbf{c}_0 \rightarrow \mathbf{c}_d | \mathbf{Y} \in \mathcal{Y}_{r'})= 1.
\end{equation*}
Therefore, the general sphere bound for the memoryless continuous channel based on the mapping function $w(\cdot)$ can be written as
\begin{equation}
	P_e = \int_{\mathcal{R}} \theta(r) g(r) dr \le \int_{\mathcal{R}} \min \left(1, \sum_{d = 1}^{n} A_d \text{Pr} (\mathbf{c}_0 \rightarrow \mathbf{c}_d | \mathbf{Y} \in \mathcal{Y}_r)\right) g(r) dr.
\end{equation}

Several GFBT-based bounds can be expressed in this form. For example, the sphere bounds for BSCs and the AWGN channel \cite{herzberg1994techniques} choose the Hamming weight and the $L_2$-norm of the noise vector as the mapping functions, respectively. Their conditional PEPs are therein derived first in the derivation of the bounds.
The Gallager region of the tangential sphere bound \cite{poltyrev1994bounds} is a circular cone whose central line passes through the origin and the transmitted signal vector.
The contour $\mathcal{Y}_r$ is the intersection of the cone and the hyperplane orthogonal to the central line, which is an $(n-2)$-dimensional sphere centered at the transmitted signal vector with radius $r$. The corresponding mapping function returns the altitude and the radius of the cone.

\subsection{PEP Conditioned on the $L_p$ Sphere}
For an AWGGN channel with noise variance $\sigma^2$, assume BPSK modulation is applied with unit signal energy. Without loss of generality, assume that the all-zero codeword $\mathbf{c}_0$ is transmitted. The channel can be described as
\begin{equation*}
	\mathbf{Y} = (2\mathbf{c}_0-1) + \mathbf{Z},
\end{equation*}
where $\mathbf{Z}$ is a random vector with each term following GGD with shape parameter $p$ and scale parameter $\alpha$.
The p.d.f. for the received vector $\mathbf{y}$ conditioned on $\mathbf{c}_0$ can be written as
\begin{equation}
	f(\mathbf{y}|\mathbf{c}_0) = \left(\frac{p}{2\alpha \Gamma(\frac{1}{p})}\right)^n e^{-\frac{1}{\alpha^p} \|\mathbf{y} - (2\mathbf{c}_0-1)\|_p^p}.
\end{equation}

Under the MLD, the log-likelihood ratio between $\mathbf{c}_0$ and a weight-$d$ codeword $\mathbf{c}_d$ for a given received vector $\mathbf{y}$ is
\begin{align}
	\log \frac{f(\mathbf{y}|\mathbf{c}_0)}{f(\mathbf{y}|\mathbf{c}_d)} &=\frac{1}{\alpha^p} \left( \|\mathbf{y} - (2\mathbf{c}_d-1)\|_p^p - \|\mathbf{y} - (2\mathbf{c}_0-1)\|_p^p \right) \nonumber \\
	&= \frac{1}{\alpha^p} \left(\|\mathbf{z}-2\|_p^p - \|\mathbf{z}\|_p^p\right).
\end{align} 
An error occurs if the log-likelihood ratio is smaller than 0.
Denote $\Delta_d$ as the decoding metric, which is expressed in terms of $Z_i$'s as
\begin{equation}
	\Delta_d  = \sum_{i=1}^d |Z_i-2|^p - |Z_i|^p.
\end{equation} 
The PEP between $\mathbf{c}_0$ and $\mathbf{c}_d$ can be denoted as $\text{Pr} (\Delta_d <0)$.

Consider $w(\mathbf{z}) = \|\mathbf{z}\|_p$ as the mapping function. 
The sphere bound for the AWGGN channel can be written as
\begin{equation}
	P_e \le \int_{0}^{\infty} \min \left(1, \sum_{d = 1}^{n} A_d \text{Pr}(\Delta_d < 0| \|\mathbf{Z}\|_p = r)\right) g(r) dr,
	\label{eq:sb_Lp}
\end{equation}
where the PEP conditioned on the $L_p$-sphere is
\begin{equation}
	\text{Pr}\left( \left. \Delta_d < 0 \right| \|\mathbf{Z}\|_p = r \right) = \underbrace{\int \cdots \int}_d I\left(\Delta_d <0\right) f\left(z_1,\cdots,z_d | \sum_{i=1}^{n}|z_i|^p = r^p\right) dz_1,\cdots dz_d,
	\label{eq:condition_PEP}
\end{equation}
and the conditional distribution can be derived based on Lemma \ref{lemma:gamma} as
\begin{align}
	f\left(z_1,\cdots,z_d | X = r^p\right) &= \frac{f(z_1,\cdots,z_d)f(\sum_{i=1}^{n}|z_i|^p=r^p|z_1,\cdots,z_d) }{f(\sum_{i=1}^{n}|z_i|^p=r^p)} \nonumber\\
	& =\frac{\prod_{i=1}^d f(z_i) f_{X}\left(r^p-\sum_{i=1}^d |z_i|^p; n-d \right)}{f_{X}(r^p;n)} \nonumber\\
	& =  \frac{p^d \Gamma(\frac{n}{p}) r^{-d}}{2^d\Gamma(\frac{n-d}{p}) \Gamma^d(\frac{1}{p})} \left(1- \frac{\sum_{i=1}^d |z_i|^p}{r^p}\right) ^{\frac{n-d}{p}-1}, \\
	& \hspace{-20pt} -r \le z_i \le r, \text{ for } i = 1,\cdots,d, \text{ and } \sum_{i=1}^d |z_i|^p \le r^p. \nonumber
\end{align}

The multiple integral in (\ref{eq:condition_PEP}) is hard to solve, except for the $p=2$ case. This dues to the spherical symmetry of the AWGN channel, which makes both the received vector and the transmitted signal vector lie on the first coordinate. At the same time, the decision boundary for the AWGN channel of the PEP is a hyperplane.
Based on these properties, Herzberg and Poltyrev find out that the conditional PEP equals the surface area ratio of the spherical cap cut out by the decision hyperplane to the whole sphere in the geometric interpretation. They derive the sphere bound for the AWGN channel in \cite{herzberg1994techniques}. But for the other choices of $p$, these kinds of properties do not hold anymore.

\subsection{Conditional PEP for the AWLN Channel}
Although there is no simple geometric interpretation for the AWLN channel, if we can derive the p.d.f. of the metric $\Delta_d$ conditioned on the radius $r$, then the conditional PEP in (\ref{eq:condition_PEP}) can be solved mathematically.

The decoding metric for $p=1$ can be written as
\begin{equation}
	\Delta_d = \sum_{i=1}^{d} |Z_i-2| - |Z_i|.
	\label{eq:delta_d}
\end{equation}

Assume the numbers of $Z_i$'s that satisfy $0\le Z_i \le2$, $Z_i >2$ and $Z_i <0$ are $d_0$, $d_2$ and $d-d_2-d_0$, respectively. Since $Z_i$'s are i.i.d., denote $D$ as the event that $0 \le Z_i \le2$ for $i = 1,2,\cdots,d_0$, $Z_i > 2$ for $i = d_0+1,\cdots, d_2+d_0$ and $Z_i < 0$ for $i = d_2+d_0+1,\cdots, d$. Then the metric can be further expressed as
\begin{equation*}
	\Delta_d = \sum_{i=1}^{d_0} Z_i -d+2d_2.
\end{equation*}

The conditional PEP in (\ref{eq:condition_PEP}) for the AWLN channel can be written as
\begin{equation}
		\text{Pr}\left(\Delta_d<0 | \|\mathbf{Z}\|_1 = r\right) = \sum_{d_0 = 0}^d \sum_{d_2 = 0}^{d-d_0} \binom{d}{d_0} \binom{d-d_0}{d_2} \text{Pr} \left(\Delta_d<0, D | r\right),
	\label{eq:condition_PEP_p1}
\end{equation}
where the joint probability for $\Delta_d < 0$ and $D$ conditioned on $r$ is 
\begin{equation}
		\text{Pr} \left(\Delta_d<0,D| r\right) = \underbrace{\int \hspace{-3pt} \cdots \hspace{-3pt} \int}_d H\left(\sum_{i=1}^{d_0}z_i - d+2d_2\right) f(z_1,\cdots,z_d, D|r) dz_1\cdots dz_d,
	\label{eq:condition_subPEP_p1}
\end{equation}
$H(\cdot)$ is the Heaviside step function defined as
\begin{equation*}
	\begin{aligned}
		H(x) = \left\{ \begin{matrix}
			&0, & x<0 \\
			&1/2, & x=0 \\
			&1, &x>0
		\end{matrix}\right.,
	\end{aligned}
\end{equation*}
and the conditional distribution of $Z_1, \cdots, Z_d$ is
\begin{align*}
	&f(z_1,\cdots, z_{d},D|r) = \frac{\Gamma(n)r^{-d}}{2^d\Gamma(n-d)}\left(1 - \frac{\sum_{i=1}^{d} |z_i|}{r}\right)^{n-d-1}, \\
	& \hspace{40pt} \left\{\begin{matrix*}[l]
		&0 \le z_i \le 2, 	&\text{ for } i = 1,2,\cdots, d_0\\
		& 2 <z_i \le r, 			&\text{ for } i  = d_0+1, \cdots, d_0+d_2\\
		& z_i < 0, 							&\text{ for } i = d_0+d_2+1, \cdots, d
	\end{matrix*} \right., \text{ and } \sum_{i=1}^d |z_i| \le r.
\end{align*}
One remark is that we assume the probability for the received vector being decoded to $\mathbf{c}_d$ when $\Delta_d = 0$ is $1/2$ as $\text{Pr}(\Delta_d=0) \neq 0$ for the AWLN channel. 

The marginal distribution of $Z_1, \cdots, Z_{d_0}$ after integrating over all $Z_i$'s that are smaller than 0 or greater than 2 is 
\begin{align*}
	&f(z_1, \cdots, z_{d_0},D|r) = \frac{\Gamma(n)r^{-d_0}}{2^d \Gamma(n-d_0)} \left(1-\frac{2d_2}{r} - \frac{\sum_{i=1}^{d_0}z_i}{r}\right)^{n-d_0-1},\\
	& \hspace{90pt} 0 \le z_i \le 2, \text{ for } i = 1, 2, \cdots, d_0, \text{ and } \sum_{i=1}^{d_0} z_i \le r-2d_2.
\end{align*}

In the following theorem, we get the p.d.f. of $\sum_{i=1}^{d_0} Z_i$ based on the above marginal distribution and derive the conditional PEP for the AWLN channel in a closed-form expression.
 
\begin{theorem}
	\label{th:sb_l1}
	Consider a binary linear block code with block length $n$ transmit over the AWLN channel with BPSK modulation. For any two codewords with Hamming distance $d$, their pairwise error probability conditioned on the $L_1$-sphere centered at one of them with radius $r$ is
	\begin{align}
		&\text{\normalfont Pr}(\Delta_d<0 |\|\mathbf{Z}\|_1 \hspace{-3pt} =r) = \frac{1}{2^d}\sum_{d_0 = 1}^d \binom{d}{d_0}  \frac{\Gamma(n)r^{1-n}}{\Gamma(d_0)\Gamma(n-d_0)} \sum_{d_2 = 0}^{d-d_0} \binom{d-d_0}{d_2} \hspace{-5pt} \sum_{m=0}^{n-d_0-1} \hspace{-5pt} \binom{n-d_0-1}{m}   \nonumber \\
		&\hspace{30pt} \cdot \sum_{l=0}^{d_0} \binom{d_0}{l} \frac{(-1)^{n+d_0-1-m-l}}{n-1-m} (\overline{x}-2l)^{n-1} \left( H(\overline{x}-2l) - \left(\frac{\underline{x}-2l}{\overline{x}-2l} \right)^{n-1-m} \hspace{-10pt} H(\underline{x}-2l)\right)  \nonumber \\
		&\hspace{230pt} +\frac{1}{2^d}\sum_{d_2=0}^{d} \binom{d}{d_2} \left(\frac{\overline{x}}{r} \right)^{n-1} H(\overline{x}) H(-\underline{x}),
		\label{eq:condtion_PEP_p1_res}
	\end{align}
	where $\Gamma(n) = (n-1)!$ is a factorial function for positive integers, $\overline{x} = r-2d_2$ and $\underline{x} = d-2d_2$.
\end{theorem}
\begin{proof}
	The proof is given in Appendix.
\end{proof}

The p.d.f. of $R$ for $p=1$ is 
\begin{equation}
	g(r) = \frac{2^{\frac{n}{2}}}{\Gamma(n) \sigma^n} r^{n-1} e^{-\frac{\sqrt{2}r}{\sigma}}, \quad r \ge 0.
	\label{eq:pdf_p1}
\end{equation}
Therefore, the sphere bound for AWLN channel can be obtained.
\begin{corollary}
	\label{coro:sb_laplace}
	The MLD error probability of any linear block code for the AWLN channel with BPSK modulation is upper bounded by
	\begin{equation}
		P_e \le \frac{1}{\Gamma(n) \alpha^n} \int_0^{\infty} \min\left(1, \sum_{d=d_{\text{\normalfont min}}}^n A_{d} \text{\normalfont Pr} \left(\Delta_d < 0 | \|\mathbf{Z}\|_1 = r \right) \right)   r^{n-1} e^{-\frac{r}{\alpha}}dr,
		\label{eq:WER_sb_p1}
	\end{equation}
	where $d_{\text{min}}$ is the minimum distance of the code and $\text{\normalfont Pr} \left(\Delta_d < 0 | \|\mathbf{Z}\|_1 = r \right)$ is given in (\ref{eq:condtion_PEP_p1_res}).
\end{corollary}

Above all, the linearity of the decoding metric helps to make the sphere bound of the AWLN channel a one-dimensional integral.
While for the other $p$ values, especially for those non-integers, the problem of how to efficiently compute the sphere bound remains open.

\section{Asymptotic IS Gain}
\label{sec:ISgain}
In this section, the efficiency of the proposed IS estimator for the AWGGN channel in terms of the sample size is investigated. We derive the asymptotic IS gain in a multiple integral form as SNR tends to infinity, where the dimension of the integral is the minimum distance of the code. 
Specifically, for the AWLN channel and the AWGN channel, based on the closed-form conditional PEPs, we derive the gains in a one-dimensional integral form, which significantly reduce the complexity of the numerical calculation.

From the variances of the MC and IS estimators in (\ref{eq:MC_var}), (\ref{eq:IS_var}) and their relationships between the relative error in (\ref{eq:SysM_relative_error}), we can get
\begin{align*}
&N_{\text{MC}} = \frac{P_e(1-P_e)}{\text{Var} \left[\hat{P}_e^{\text{MC}}\right]} = \frac{1}{\kappa^2 P_e} - \frac{1}{\kappa^2}, \\
& N_{\text{IS}} = \frac{\int_{0}^{\infty} \theta(r) \frac{g^2(r)}{g^*(r)} dr -P_e^2 }{\text{Var} \left[\hat{P}_e^{\text{IS}}\right]} = \frac{\int_{0}^{\infty} \theta(r) \frac{g^2(r)}{g^*(r)} dr }{\kappa^2 P_e^2} - \frac{1}{\kappa^2},
\end{align*}

The IS gain is defined as the ratio of the variances of the MC and IS estimators given sample size $N$ \cite{xia2003importance}, which is equivalent to the ratio of the generated sample size under the same relative error $\kappa$. 
\begin{equation}
\gamma \triangleq \frac{N_{\text{MC}}}{N_{\text{IS}}} = \frac{P_e - P_e^2}{\int_{0}^{\infty} \theta(r) \frac{g^2(r)}{g^*(r)} dr - P_e^2}.
\label{eq:ISgain_init}
\end{equation}

By substituting the optimal IS distribution (\ref{eq:opt_IS_r}) into (\ref{eq:ISgain_init}), we can get
\begin{equation}
\gamma = \frac{P_e - P_e^2}{\left(\int_{0}^{\infty} \sqrt{\theta(r)}g(r) dr \right)^2 - P_e^2}.
\label{eq:ISgain}
\end{equation}

We consider the asymptotic IS gain as SNR tends to infinity. The major problem here is to represent $P_e$ and $\theta(r)$ in (\ref{eq:ISgain}) in terms of some tabulated coding parameters such as $d_{\text{min}}$ for the asymptotic case.

The union bound indicates that the WER is upper bounded by
\begin{equation}
	P_e \le \sum_{d=d_{\text{min}}}^{n} A_d \text{Pr}(\mathbf{c}_0 \rightarrow \mathbf{c}_d).
\end{equation}
Since the PEP decays exponentially in terms of the Hamming weight $d$, the WER is dominated by the weight-$d_{\text{min}}$ codewords for the asymptotic case
\begin{equation}
	P_e \simeq A_{d_{\text{min}}} \text{Pr} (\mathbf{c}_0 \rightarrow \mathbf{c}_{d_{\text{min}}}).
	\label{eq:asym_WER}
\end{equation}

According to the unified PEP calculation formula in \cite{biglieri1998computing}, the exact PEP can be written as
\begin{equation}
\text{Pr}(\mathbf{c}_0 \rightarrow \mathbf{c}_d) = P(\Delta_d < 0) =  \frac{1}{2\pi j} \int_{c-j\infty}^{c+j\infty}  \frac{\Phi(s)}{s} ds
\label{eq:uniPEP}
\end{equation}
where
\begin{equation*}
\Phi(s) = E\left[e^{-s\Delta_d}\right] = \prod_{i=1}^d E \left[e^{-s\left(|Z_i-2|^p-|Z_i|^p\right)}\right] = \left( \int_{-\infty}^{\infty} e^{-s\left(|z-2|^p-|z|^p\right)} f(z) d z \right)^d
\end{equation*}
is the (two-sided) Laplace transform of the p.d.f. of $\Delta_d$.

According to \cite{biglieri1998computing}, the integral in (\ref{eq:uniPEP}) can be numerically calculated using Gaussian quadratic rule with an even number $m$ of nodes,
\begin{equation}
\text{Pr}(\Delta_d < 0) =\frac{1}{m} \sum_{i=1}^{m/2} \left(\text{Re}\left[\Phi(c+jc\tau_i)\right]+ \tau_i \text{Im}\left[\Phi(c+jc\tau_i)\right] \right)- \epsilon_m,
\end{equation}
where $\tau_i = \tan \left(\frac{(2i-1)\pi}{2m}\right)$, and the error term $\epsilon_m$ vanishes as $m \rightarrow \infty$. The value of $c$ affects the number of nodes necessary to achieve a preassigned accuracy.

Therefore, the asymptotic WER can be numerically evaluated. In the following lemma, we derive the asymptotic error ratio in a multiple integral form.
\begin{lemma}
	Consider a linear block code with minimum distance $d_{\text{min}}$ transmitted through an AWGGN channel with shape parameter $p \ge1 $. Assume the BPSK modulation is applied. As SNR tends to infinity, the asymptotic error ratio for $r \in [d_{\text{min}}^{\frac{1}{p}}, (d_{\text{min}}+1)^{\frac{1}{p}})$ can be approximated by
	\begin{align}
		\theta(r) \simeq A_d \cdot \frac{p^d \Gamma(\frac{n}{p}) r^{-d}}{2^d\Gamma(\frac{n-d}{p}) \Gamma^d(\frac{1}{p})} \underbrace{\int \cdots \int}_{d} I\left( \Delta_d <0 \right) \left(1- \frac{\sum_{i=1}^d |z_i|^p}{r^p}\right) ^{\frac{n-d}{p}-1}dz_1 \cdots d z_d,
		\label{eq:asym_theta}
	\end{align}
	where $d$ refers to $d_{\text{min}}$ for simplicity, and $\Delta_d$ is defined in (\ref{eq:delta_d}).
\end{lemma}
\begin{proof}
	According to the union bound, the error ratio for the AWGGN channel is upper bounded by
	\begin{equation*}
		\theta(r) \le \sum_{d = d_{\text{min}}}^{n} A_d \text{Pr} (\Delta_d <0 | \|\mathbf{Z}\|_p = r),
	\end{equation*}
	where the conditional PEP is provided in (\ref{eq:condition_PEP}).
	
	The p.d.f. $g(r)$ of the random variable $R = \|\mathbf{Z}\|_p$ is shown in (\ref{eq:pdf_R}). Its derivative can be derived as $g'(r) = g(r)\left(\frac{n-1}{r} - \frac{pr^{p-1}}{\alpha^p}\right)$. In most cases, $n \gg p $ holds, which means $g(r)$ decays exponentially as $r > \alpha \left(\frac{n-1}{p}\right)^{\frac{1}{p}}$. Therefore, for the asymptotic case (i.e., $\alpha \rightarrow 0$),  the errors within spheres $\mathcal{Y}_r$'s whose radii $r$ are very close to the packing radius dominate the error performance. 
	
	On the other hand, according to the Minkowski inequality, $\|\mathbf{z}\|_p + \|2\left(\mathbf{c}_0 - \mathbf{c}_d\right) - \mathbf{z}\|_p \ge \|2\left(\mathbf{c}_0 - \mathbf{c}_d\right)\|_p = 2d^{\frac{1}{p}}$ holds when $p \ge 1$. The minimum value of $\|\mathbf{z}\|_p$ when the metric difference equals 0 (i.e., $ \|\mathbf{z}\|_p = \|\mathbf{z} - 2\left(\mathbf{c}_0 - \mathbf{c}_d\right)\|_p$) is $\|\mathbf{z}\|_p = d^{\frac{1}{p}}$. Hence, the packing radius of the BPSK-modulated code transmitted over AWGGN channel with MLD is $r_{\text{pack}} = d_{\text{min}}^{\frac{1}{p}}$. For $r \in [d_{\text{min}}^{\frac{1}{p}}, (d_{\text{min}}+1)^{\frac{1}{p}})$, there only exist the errors wrongly decoded to the weight-$d_{\text{min}}$ codewords. Therefore, the following approximation holds for the asymptotic case.
	\begin{equation}
		\theta(r) \simeq A_{d_{\text{min}}}\text{Pr}\left(\Delta_{d_{\text{min}}} < 0 | \|\mathbf{Z}\|_p = r\right).
		\label{eq:theta_asym}
	\end{equation}
	By substituting (\ref{eq:condition_PEP}) into (\ref{eq:theta_asym}), the asymptotic error ratio can be derived.
\end{proof}

Since the WER satisfies $P_e \ll 1$ for the asymptotic case, the IS gain can be approximated by
\begin{equation}
	\gamma \simeq \frac{P_e}{\left(\int_{0}^{\infty} \sqrt{\theta(r)}g(r) dr \right)^2}.
	\label{eq:ISgain_approx}
\end{equation}
By bringing the asymptotic WER (\ref{eq:asym_WER}) and error ratio (\ref{eq:asym_theta}) into the above equation, we can derive the asymptotic IS gain that only depends on the minimum distance.

As the minimum distance increases, the number of dimensions of the multiple integral in (\ref{eq:asym_theta}) increases, which makes the calculation of the asymptotic IS gain more and more difficult. In the rest of the section, we focus on the AWLN and AWGN channels and derive the mathematical expressions of the asymptotic IS gains in a one-dimensional integral form, which significantly reduce the complexity of the numerical calculation.

\subsection{AWLN Channel, $p=1$}
For the case of $p=1$, according to the PEP for the AWLN channel provided in \cite{marks1978detection,shao2012investigation}, the asymptotic WER can be expressed as
\begin{align}
	&P_e \simeq 2^{-d}A_d \sum_{d_0=1}^{d} \binom{d}{d_0} \sum_{d_2=0}^{d-d_0} \binom{d-d_0}{d_2} \sum_{l =0}^{d_0} \binom{d_0}{l}(-1)^l e^{-\frac{2\sqrt{2}}{\sigma}(d_2+l)} \nonumber \\  
	& \hspace{30pt} \cdot  \Gamma\left( \min\left(0,\frac{\sqrt{2}}{\sigma}(d-2d_2-2l) \right), d_0\right) + 2^{-d} A_d\sum_{d_2=0}^{d} \binom{d}{d_2} e^{-\frac{2\sqrt{2}}{\sigma}d_2} H(2d_2-d),
\end{align}
where $d$ refers to $d_{\text{min}}$ for simplicity and 
\begin{equation*}
	\Gamma(x,s) = \frac{1}{\Gamma(s)} \int_x^{\infty} t^{s-1} e^{-t} dt
\end{equation*}
is the regularized upper incomplete Gamma function.

The asymptotic error ratio can be derived by substituting (\ref{eq:condtion_PEP_p1_res}) into (\ref{eq:theta_asym}),
\begin{align}
	&\theta(r) \simeq 2^{-d} r^{1-n}A_{d}\sum_{d_0 = 0}^d \binom{d}{d_0}  \frac{\Gamma(n)}{\Gamma(d_0)\Gamma(n-d_0)} \sum_{d_2 = 0}^{d-d_0} \binom{d-d_0}{d_2} \sum_{m=0}^{n-d_0-1}\binom{n-d_0-1}{m}  \nonumber \\
	&\hspace{100pt} \cdot \sum_{l=0}^{d_0} \binom{d_0}{l} \frac{(-1)^{n+d_0-1-m-l}}{n-1-m}  \left[(r-2d_2-2l)^{n-1} H(r-2d_2-2l) \right.  \nonumber \\
	&\hspace{135pt} \left.- (r-2d_2-2l)^{m} (d-2d_2-2l)^{n-1-m}H(d-2d_2-2l)\right] \nonumber \\
	&\hspace{110pt} + 2^{-d}r^{1-n}A_d\sum_{d_2=0}^{d} \binom{d}{d_2} \left(r-2d_2\right)^{n-1} H(r-2d_2) H(2d_2-d),
	\label{eq:theta_asym_p1}
\end{align}
where $d$ refers to $d_{\text{min}}$ for simplicity.

Consequently, by bringing the above equations and the p.d.f. of $R = \|\mathbf{Z}\|_1$ provided in (\ref{eq:pdf_p1}) into (\ref{eq:ISgain_approx}), the asymptotic IS gain for the AWLN channel can be derived.

\subsection{AWGN channel, $p=2$}

For the case of $p=2$, the asymptotic WER for the AWGN channel is
\begin{equation}
	P_e \simeq A_{d_{\text{min}}} Q\left(\frac{\sqrt{d_{\text{min}}}}{\sigma}\right),
\end{equation}
where 
\begin{equation*}
	Q(x) = \frac{1}{\sqrt{2\pi}} \int_{x}^{\infty} e^{-\frac{x^2}{2}} dx.
\end{equation*}

The conditional PEP for the AWGN channel is provided in \cite{herzberg1994techniques} as 
\begin{equation}
		\text{Pr}\left(\left.\Delta_d <0 \right| \|\mathbf{Z}\| = r\right)  = \frac{\Gamma(\frac{n}{2})}{\sqrt{\pi}\Gamma(\frac{n-1}{2})} \int_{\frac{\sqrt{d}}{r}}^{1} \left(1 - x^2\right)^{\frac{n-3}{2}} dx. 
	\label{eq:conditional_prob_p2}
\end{equation}

By changing the variable $\cos \phi =x$ and substituting (\ref{eq:conditional_prob_p2}) into (\ref{eq:theta_asym}), the asymptotic error ratio can be written as
\begin{equation}
	\theta(r) \simeq A_{d_{\text{min}}} \cdot\frac{\Gamma\left(\frac{n}{2}\right)}{\sqrt{\pi} \Gamma\left(\frac{n-1}{2}\right)} \int_0^{\arccos\left(\frac{\sqrt{d_{\text{min}}}}{r}\right)} \sin^{n-2} \phi d \phi = \frac{A_{d_{\text{min}}}}{2} I_{1-\frac{d_{\text{min}}}{r^2}} \left( \frac{n-1}{2}, \frac{1}{2}\right) ,
\end{equation}
where 
\begin{equation*}
	I_{x} \left(a, b\right) =\frac{\Gamma(a+b)}{\Gamma(a)\Gamma(b)} \int_0^x t^{a-1} (1-t)^{b-1} dt
\end{equation*}
is the regularized incomplete beta function.

The random variable $R = \|\mathbf{Z}\|$ follows scaled chi distribution with p.d.f.
\begin{equation} 
	g(r)  = \frac{1}{2^{\frac{n}{2}-1} \sigma^n \Gamma(\frac{n}{2})} r^{n-1} e^{-\frac{r^2}{2 \sigma^2}}, \quad r \ge 0.
\end{equation}

Therefore, by substituting the asymptotic WER, the error ratio and the p.d.f. into (\ref{eq:ISgain_approx}), the asymptotic IS gain for AWGN channel is
\begin{equation}
		\gamma \simeq \frac{2^{n-1} \sigma^{2n} \Gamma^2(\frac{n}{2}) Q\left(\frac{\sqrt{d_{\text{min}}}}{\sigma}\right)}{\left( {\displaystyle \int_{\sqrt{d_{\text{min}}}}^{\infty}} I^{\frac{1}{2}}_{1-\frac{d_{\text{min}}}{r^2}} \left( \frac{n-1}{2}, \frac{1}{2}\right) r^{n-1} e^{-\frac{r^2}{2\sigma^2}}dr \right)^2 }.
		\label{eq:asym_IS_gain_p2}
\end{equation}

\section{Simulation Results}
\label{sec:results}
In this section, the derived sphere bound for the AWLN channel is verified through simulation. The efficiency of the proposed IS estimator over the AWGGN channel is shown. Several different coding schemes and shape parameters are applied to show the generality of the method. In addition, the accuracy of the derived asymptotic IS gains compared to the simulated IS gains are provided.

\subsection{Sphere Bound for AWLN Channel}
We first present some examples of the derived sphere bound (\ref{eq:WER_sb_p1}) for the AWLN channel. The primitive BCH codes with parameters $(15,7)$ and $(31,11)$ are considered. In Table \ref{table:WEF_BCH}, the weight distributions of these codes, which have been tabulated in \cite{teradaWD}, are provided.

\begin{table}[!h]
	\centering
	\caption{The weight distributions of primitive BCH codes}
	\label{table:WEF_BCH}
	\begin{tabular}{|l| *{8}{c|}}
		\hline 
		\multirow{2}{*}{(15,7)} & $d$	&5		& 6 & 7 	& 8   & 9 & 10 & 15 \\ \hhline{~--------}
		&$A_d$ & 18	& 30 	& 15 	& 15 & 30 & 18 & 1  \\ \hline \hline
		\multirow{2}{*}{(31,11)} & $d$ & 11 & 12 & 15 & 16 & 19 & 20 & 31 \\ \hhline{~--------}
		& $A_d$ & 186 & 310 & 527 & 527 &  310 & 186 & 1 \\ \hline
	\end{tabular}
\end{table}
In Fig. \ref{fig:BCH_SB}, we compare the derived sphere bound in Corollary \ref{coro:sb_laplace} with the union bound exhibited in \cite{marks1978detection} and the MLD simulation results of the codes. One can see that the derived sphere bound is asymptotically tight as SNR tends to infinity and tightens the union bound especially for low SNR, where the union bound even exceeds the unity and becomes useless.

\begin{figure}[!h]
	\centering
	\subfigure[BCH (15,7)]{
		\includegraphics[width=3in]{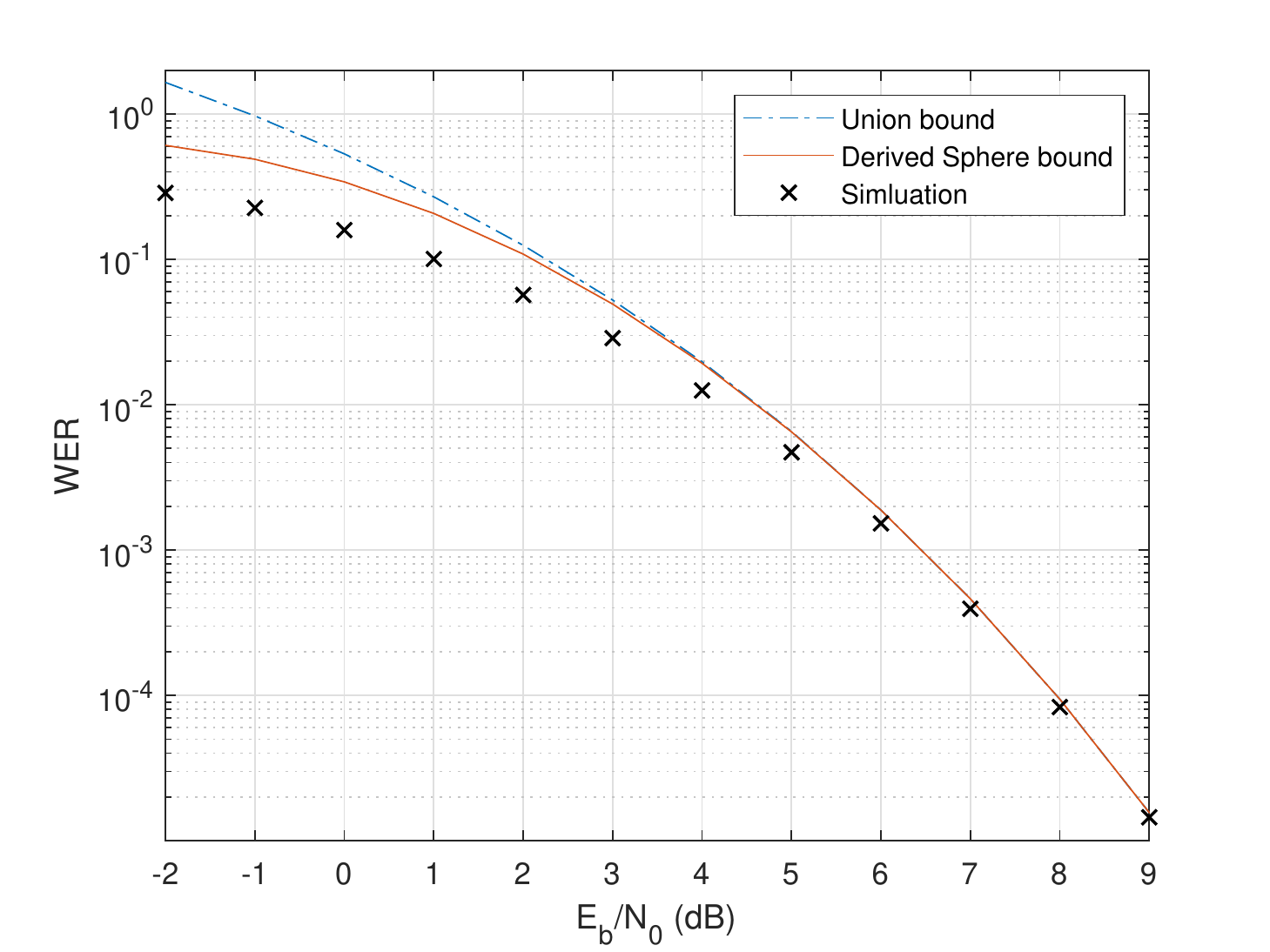} }
	\subfigure[BCH (31,11)]{
		\includegraphics[width=3in]{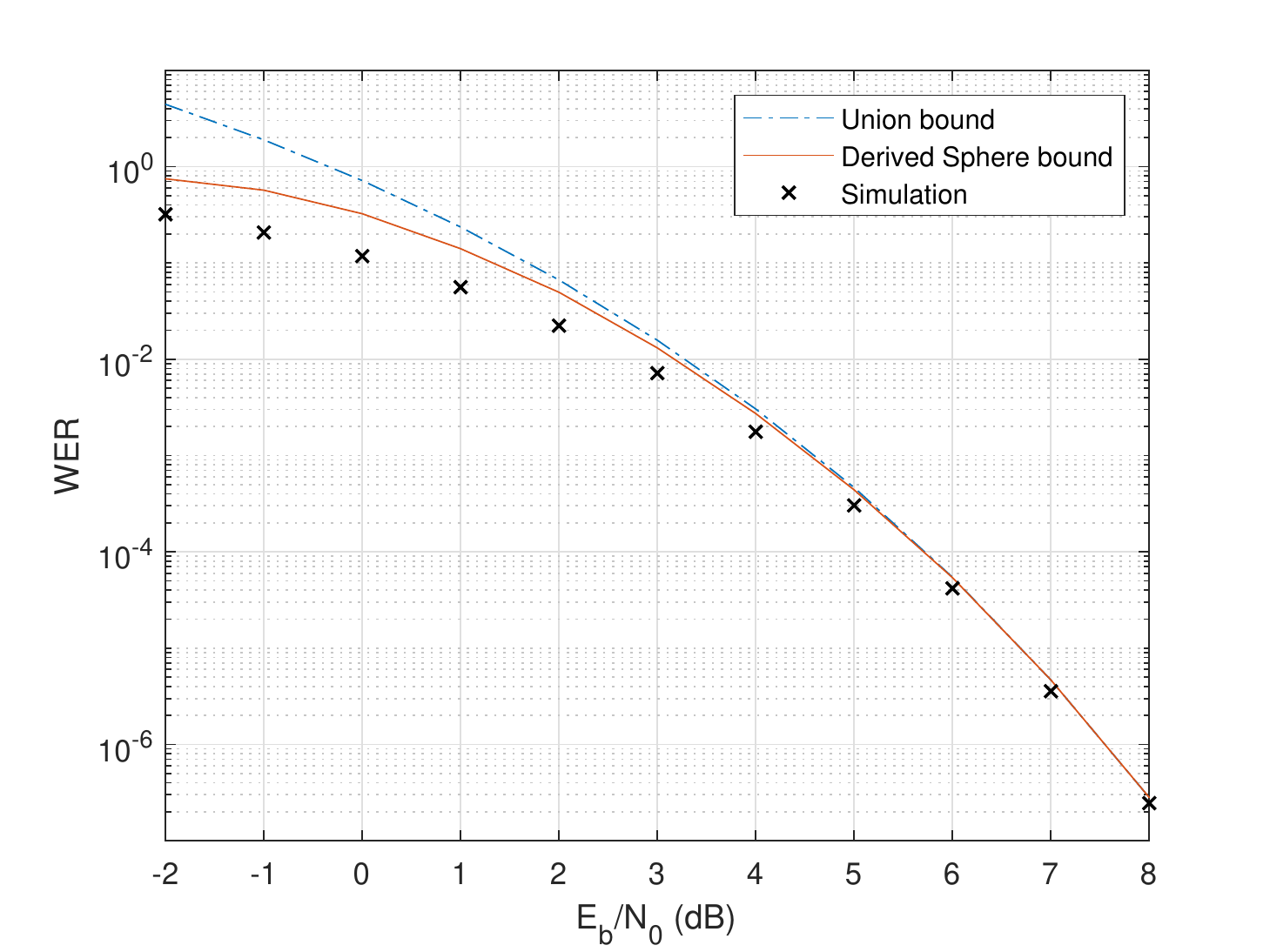} }
	\caption{The MLD performance comparison of the derived sphere bound, the union bound \cite{marks1978detection} and the simulation results of the (15,7) and (31,11) BCH codes for the AWLN channel.}
	\label{fig:BCH_SB}
\end{figure}


\subsection{Importance Sampling}
We use several examples to show the performance of the proposed IS estimator for the AWGGN channel.
Different shape parameters $p$ for the channel are considered to show the generality of the method on different application scenarios.
The accuracy of the derived asymptotic IS gain in predicting the sample size saved compared to the MC method is verified through simulation. As the derivation of the gain requires the minimum distance, the codes whose $d_\text{min}$ have already been tabulated, such as the BCH codes and the EG-LDPC codes \cite{kou2001low}, are considered.

Since there are no errors inside the packing sphere, the range for the sample generation in the radius domain is set as $r_{\text{min}} = d_{\text{min}}^{\frac{1}{p}} $ and $r_{\text{max}} = 5 d_{\text{min}}^{\frac{1}{p}} $ empirically, and the number of shells is set as $m = 500$. 
Besides, we set the minimum number of sample size as $N_{\text{min}} = 500$ and the step size for the update frequency as $N_{\text{step}} = 100$.
The stopping criterion on the relative error of the simulation is set as $\kappa = 0.1$.

\begin{figure}[!h]
	\centering
	\subfigure[WER]{
		\includegraphics[width=3in]{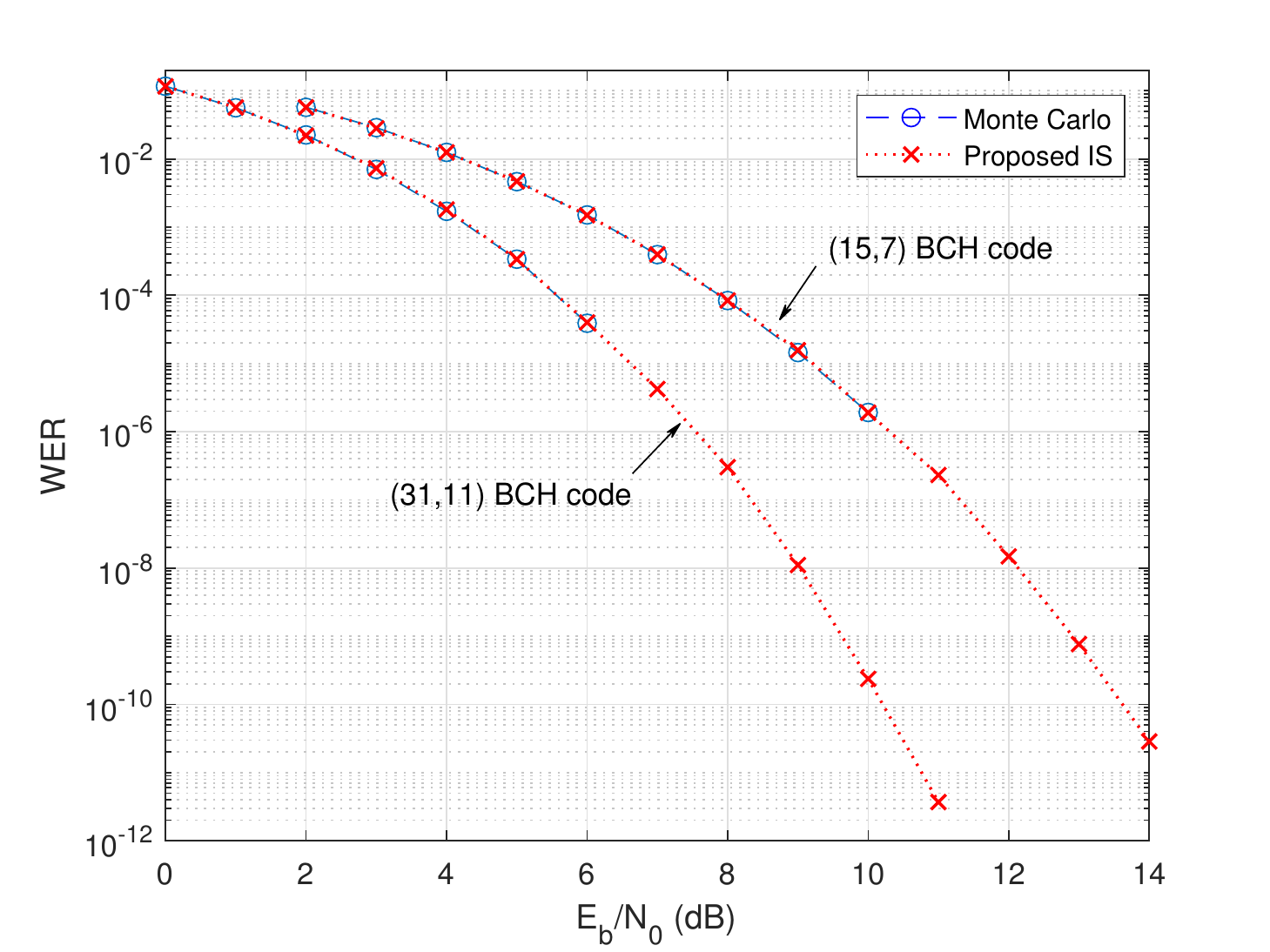} }
	\subfigure[IS gain]{
		\includegraphics[width=3in]{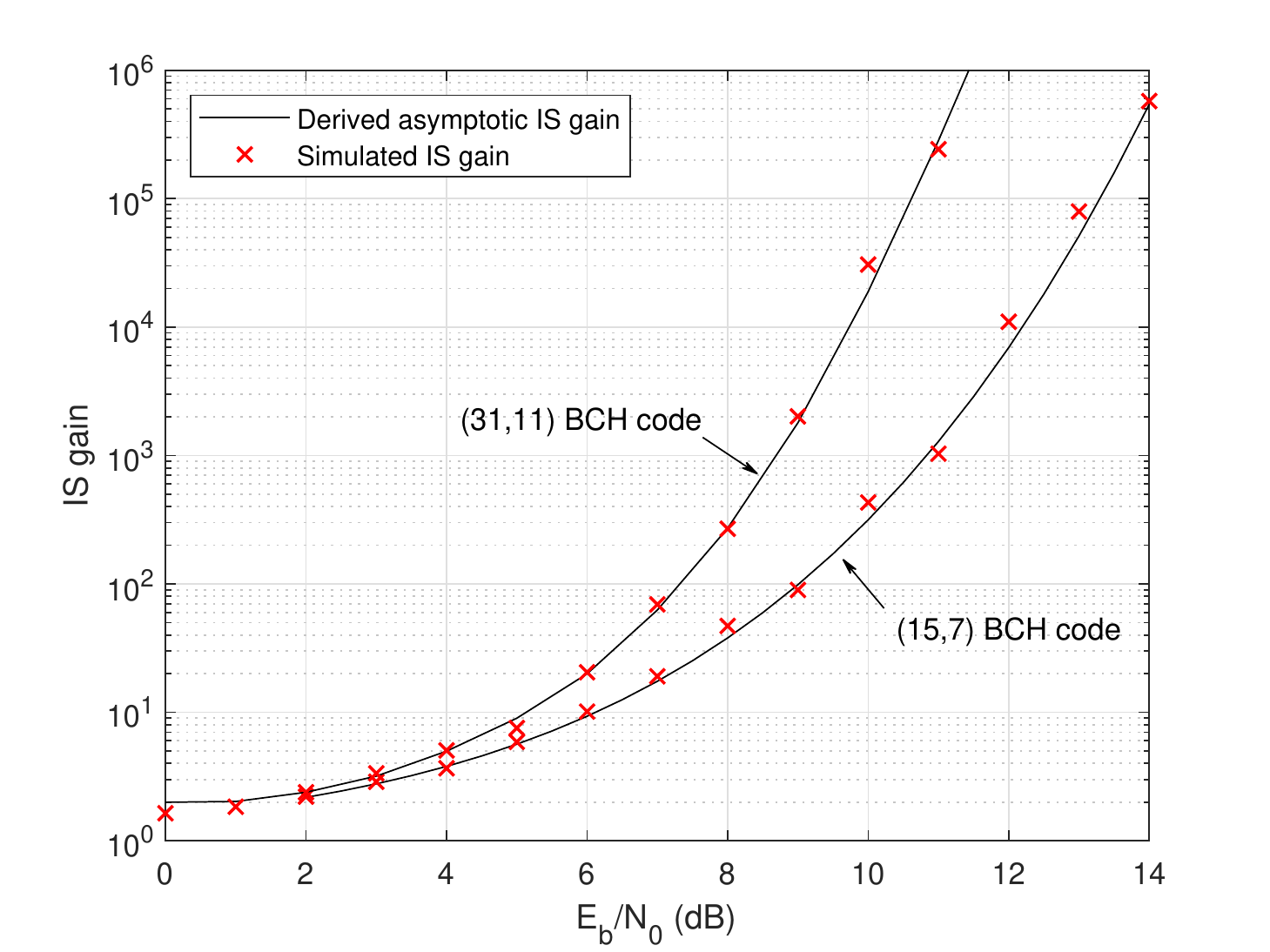} }
	\caption{The MLD simulation results of the (15,7) and (31,11) BCH codes using the proposed IS algorithm for the AWLN channel. (a) WER compared with the MC simulation; (b) Simulated IS gain compared with the derived asymptotic IS gain.}
	\label{fig:BCH_IS_p1}
\end{figure}

The WER performances and the IS gains of the proposed IS estimator using the BCH $(15,7)$ and $(31,11)$ codes for the AWLN channel with MLD are shown in Fig. \ref{fig:BCH_IS_p1}. The curves in Fig. \ref{fig:BCH_IS_p1} (a) evidently show that the IS estimates coincide with their MC counterparts for WER. 
Furthermore, the proposed IS estimator can work in the high SNR region, where the time cost of the MC method is unaffordable.
In Fig. \ref{fig:BCH_IS_p1} (b), the simulated IS gains, as well as the derived asymptotic IS gains in (\ref{eq:asym_IS_gain_p2}) for the applied codes, w.r.t. $E_b/N_0$ in dB are reported. The number of samples required for the MC method in high SNR are approximated by $100/P_e$. 
It can be noticed that the IS gains grow exponentially as SNR increases, and the sample size saved by using the proposed method can be significantly large at high SNR. Take the BCH $(31, 11)$ code as an example, it takes $4.6\times 10^6$ samples and $13$ minutes for the proposed IS estimator to estimate the WER of $10^{-8}$ at $9$ dB. As a comparison, the MC method requires around $10^{10}$ samples for an estimate at the same accuracy, which takes almost 20 days.
Besides, we can observe that the derived asymptotic IS gain fits the simulated one pretty well, which means one can predict the sample size saving using the proposed IS estimator before the simulation.

For the AWGN channel, we use the EG-LDPC codes with parameters $(15,7)$ and $(63,37)$ provided in \cite{kou2001low} to verify the efficiency of our method. The minimum distances of the codes are therein provided, which are 5 and 9, respectively. The codes are decoded by the sum-product algorithm with a maximum iteration number of $50$. 
The WER results and the IS gains of the proposed IS estimator compared to the MC method are shown in Fig. \ref{fig:EG-LDPC1_IS_p2} (a) and (b), respectively. Similar as the cases for the AWLN channel, the proposed method shows a significant advantage in terms of the efficiency in high SNR. 
\begin{figure}[!h]
	\centering
	\subfigure[WER]{
		\includegraphics[width=3in]{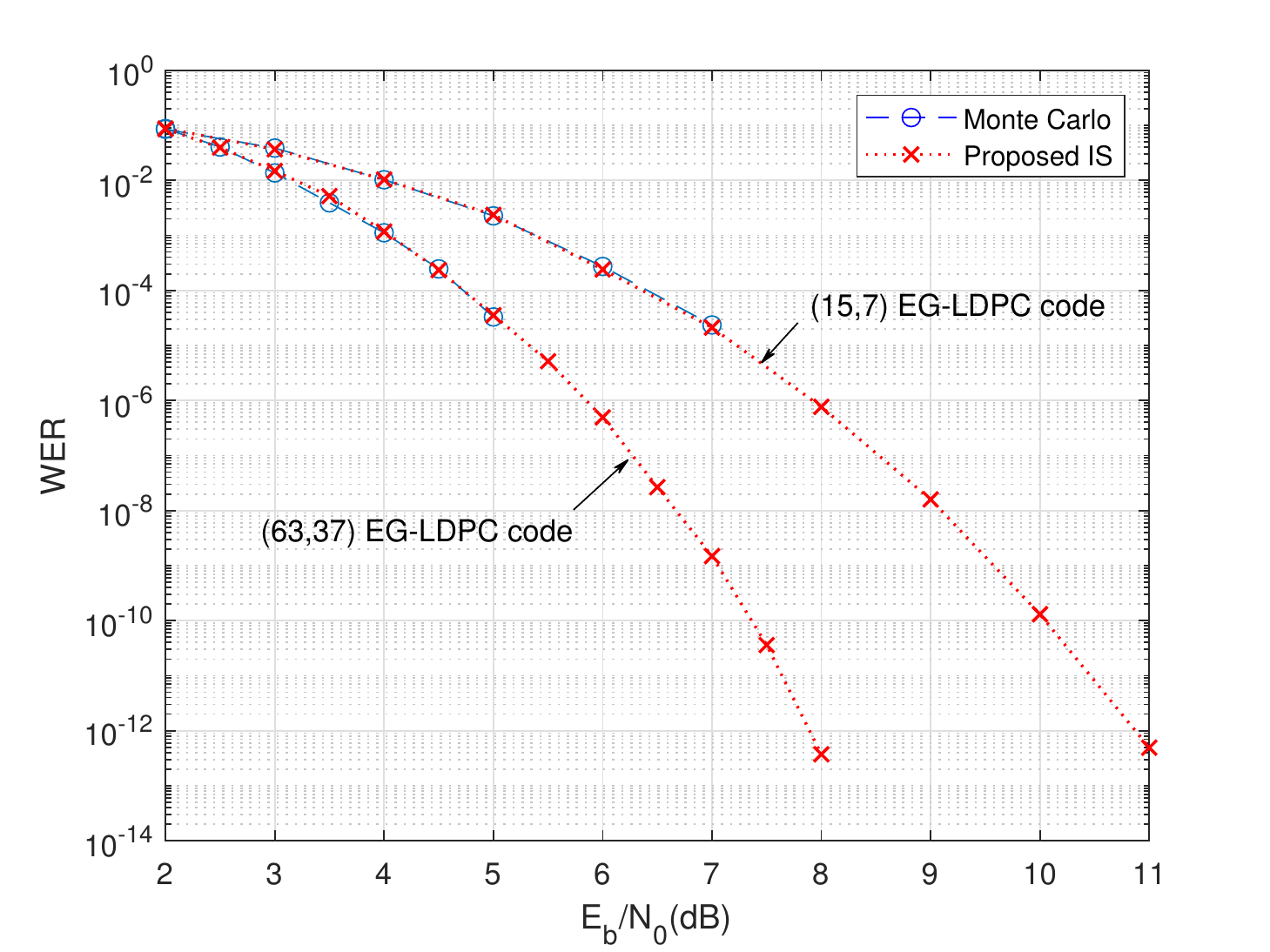} }
	\subfigure[IS gain]{
		\includegraphics[width=3in]{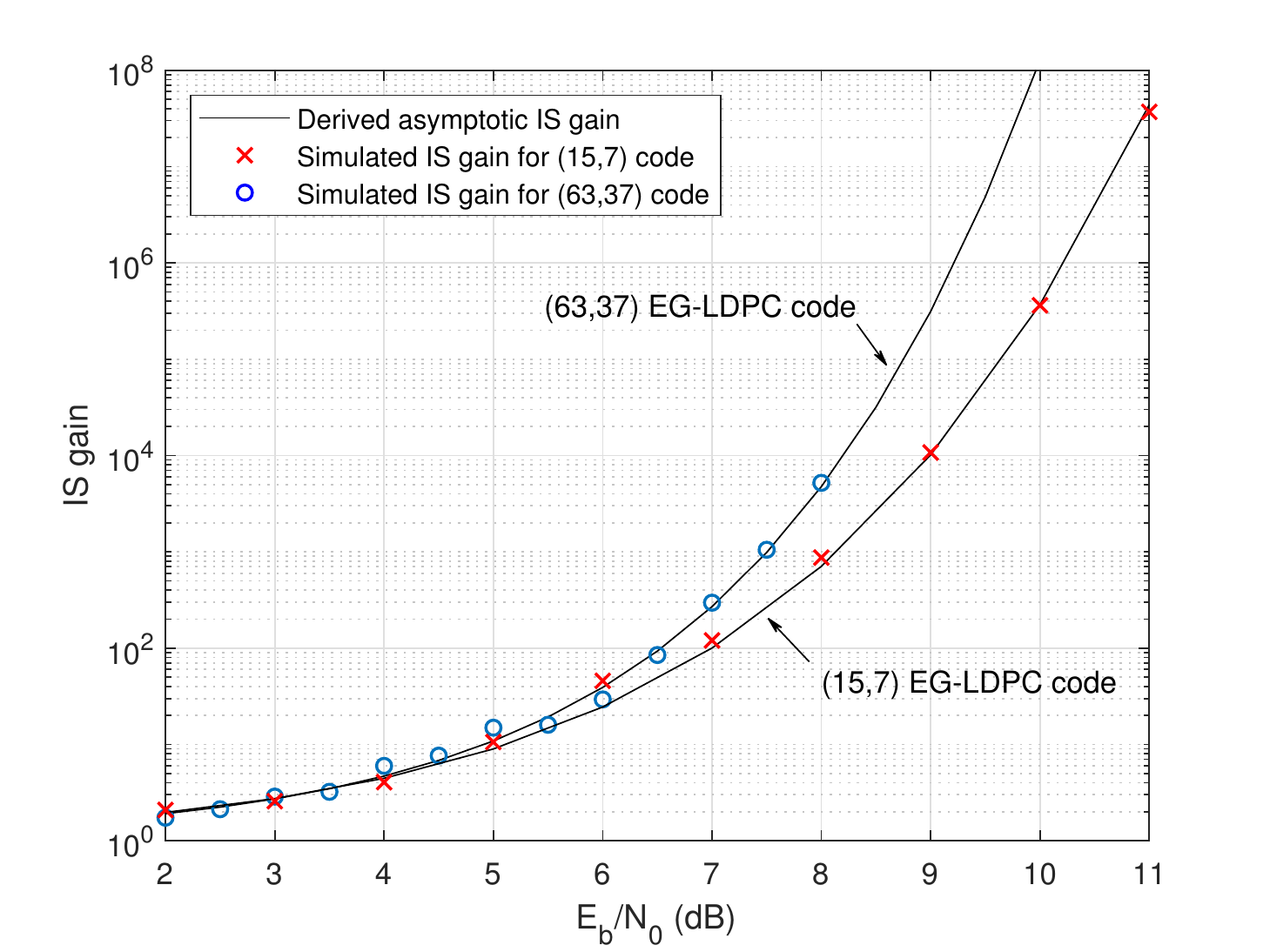} }
	\caption{ The simulation results of the (15,7) and (63,37) EG-LDPC codes \cite{kou2001low} using the proposed IS algorithm for the AWGN channel decoded by the sum-product algorithm. (a) WER compared with the MC simulation; (b) Simulated IS gain compared with the proposed asymptotic IS gain.}
	\label{fig:EG-LDPC1_IS_p2}
\end{figure}

For comparison, we also apply our method to MacKay's (96,50) LDPC code \cite{MacKayCodes}, which is used in \cite{xia2003importance} and \cite{holzlohner2005evaluation} as an example to show the efficiency of their methods. 
As the minimum distance of the code is unavailable, only the simulated IS gain is shown in  Fig. \ref{fig:MacKay} (b).
For $E_b/N_0 = 9$ dB, totally $8.9\times10^9$ samples are generated for the proposed IS method, which corresponds to an IS gain around $\gamma = 2.8\times10^3$. Compared to the method introduced in \cite{xia2003importance}, our IS estimator can achieve a three times larger IS gain. 
Although the DAIS method in \cite{holzlohner2005evaluation} also shows a significant gain, the empirical stopping criterion they used makes it hard to make a fair comparison.

\begin{figure}[!h]
	\centering
	\subfigure[WER]{
		\includegraphics[width=.45\linewidth]{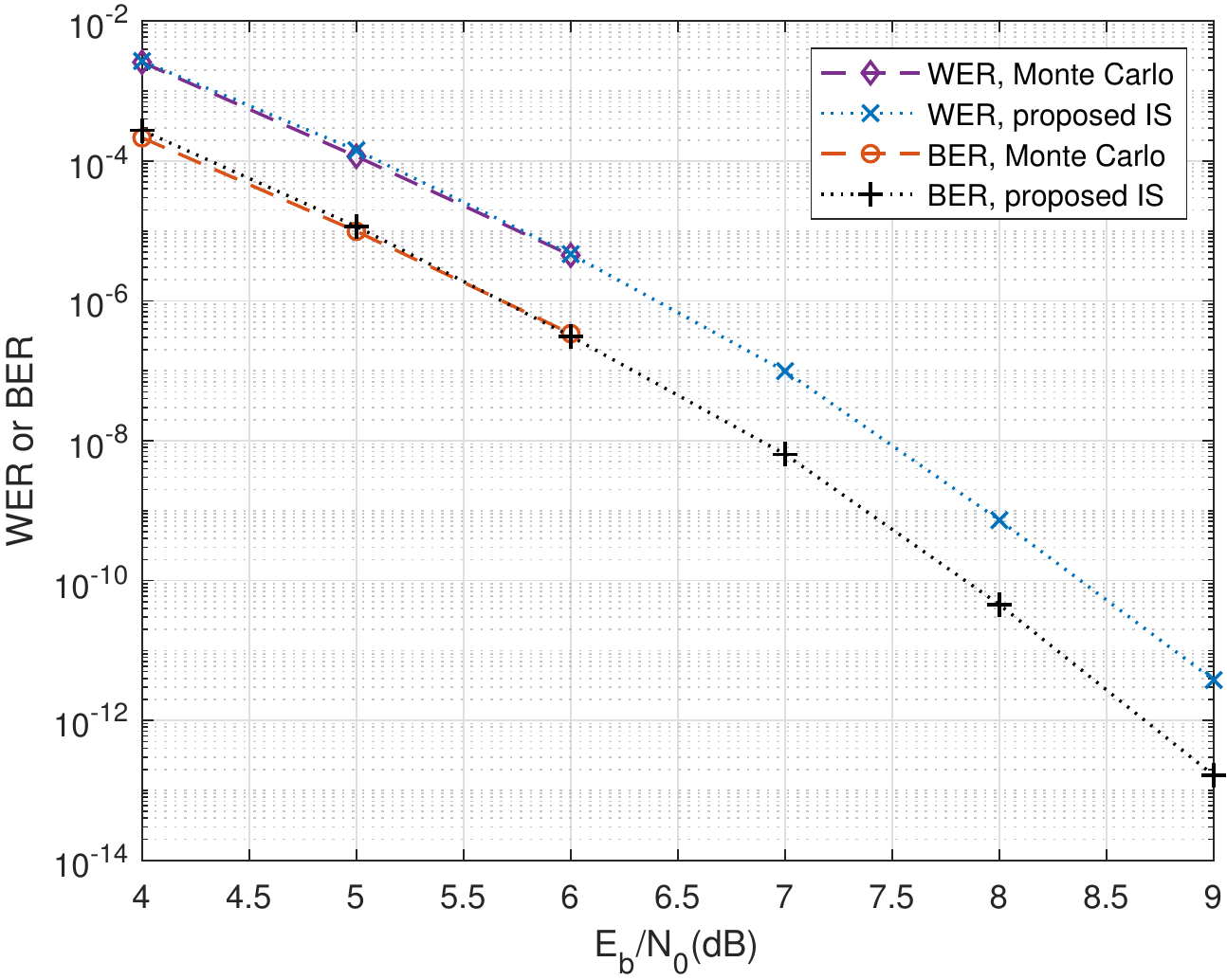} }
	\subfigure[IS gain]{
		\includegraphics[width=.45\linewidth]{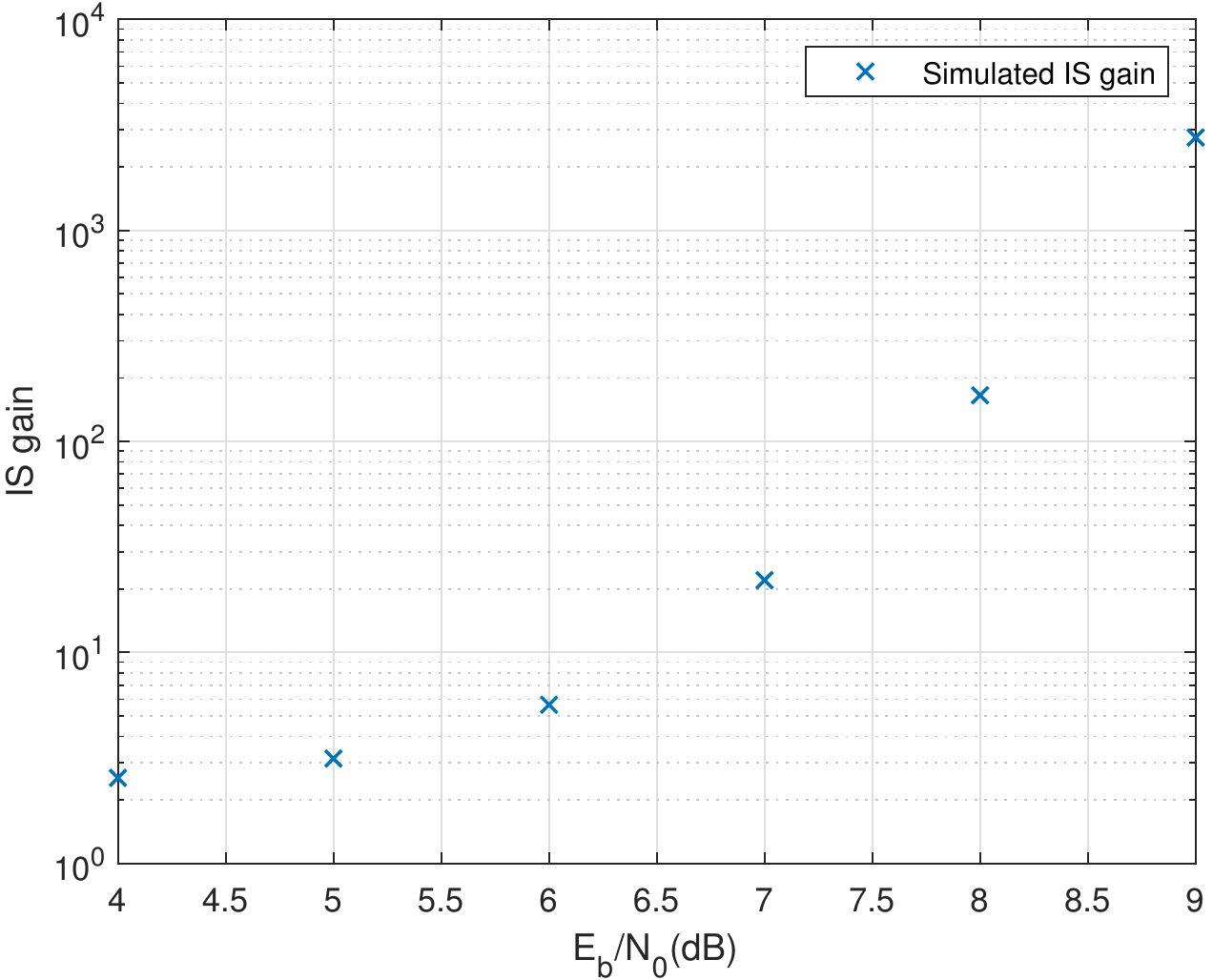} }
	\caption{Implementation of the proposed IS estimator using (96,50) MacKay's code \cite{MacKayCodes}. (a) WER and BER comparisons between the MC and the proposed IS method; (b) Simulated IS gain of the proposed IS method.}
	\label{fig:MacKay}
\end{figure}

We also consider the AWGGN channels with $p = 1.6$ and $2.8$ as examples to show the generality of the proposed IS estimator. The $p$ values are chosen according to the underwater acoustic channel \cite{banerjee2013underwater}. The performances of the BCH $(15,7)$ code with MLD are shown in Fig. \ref{fig:BCH_IS_p}. For reference, the results for the AWLN and AWGN channel are also shown. 
\begin{figure}[!h]
	\centering
	\subfigure[WER]{
		\includegraphics[width=3in]{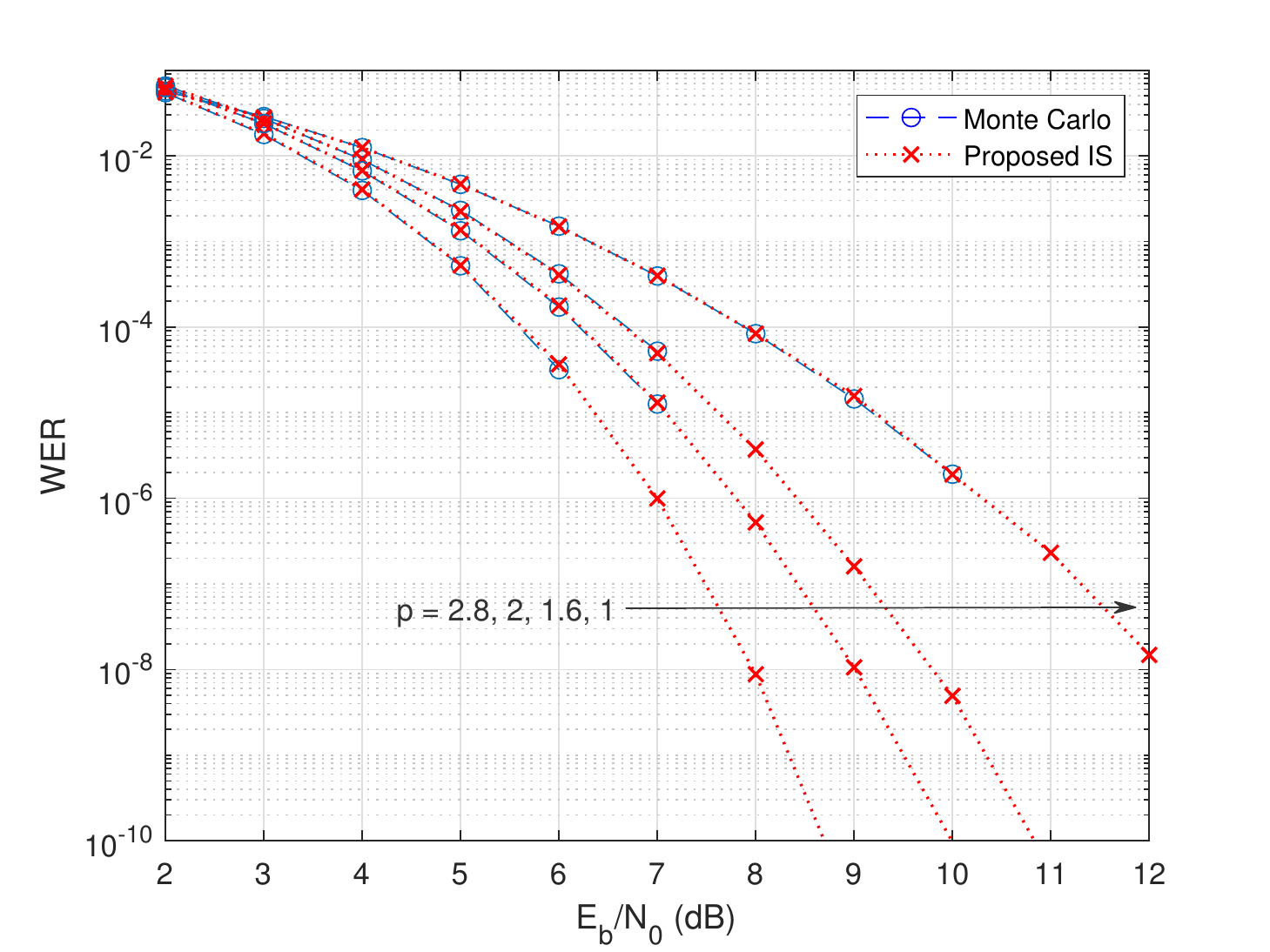} }
	\subfigure[IS gain]{
		\includegraphics[width=3in]{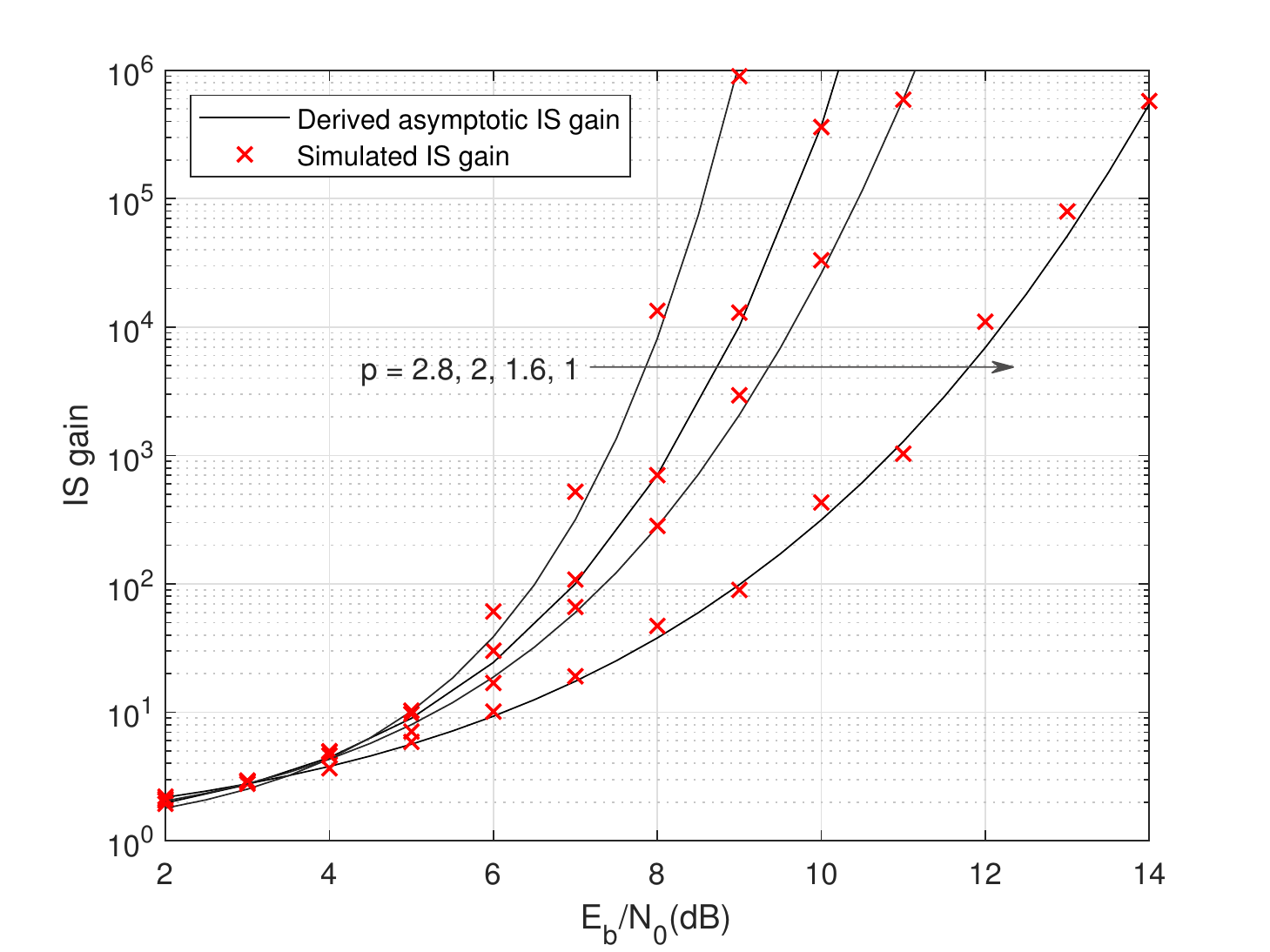} }
	\caption{The MLD simulation results of the (15,7) BCH code using the proposed IS algorithm for the AWGGN channel with various $p$ values. (a) WER compared with the MC simulation; (b) Simulated IS gain compared with the proposed asymptotic IS gain.}
	\label{fig:BCH_IS_p}
\end{figure}
As we can see from Fig. \ref{fig:BCH_IS_p} (b), the proposed IS estimator shows advantage in terms of efficiency for all the applied scenarios, and the derived asymptotic IS gain predicts its simulated counterpart accurately. 
These derived IS gains for the non-integer values of $p$ are obtained by numerical integration.
As the minimum distance of the code increases, it becomes infeasible to numerically calculate its asymptotic IS gain. Therefore, we can use the IS gains of the AWLN and AWGN cases to roughly determine the range of the gains for non-integer $p$ cases.

\section{Conclusion}
\label{sec:conclusion}
In this paper, the problem of efficiently evaluating the error performance of linear block codes over an AWGGN channel is investigated. 
For the memoryless continuous channels, we present a general framework to designing IS estimators by defining a random variable as a function of the $n$-dimensional received vector and deriving the optimal IS distribution for it. 
Specifically for the AWGGN channel, we choose the $L_p$-norm as the mapping function and propose a minimum-variance IS estimator. 
We investigate the essential role that the conditional PEP plays in the derivation of the sphere bound and derive the PEP conditioned on the $L_1$-sphere for the AWLN channel in a closed-form expression.
By choosing the $L_1$-sphere centered at the transmitted signal vector as the Gallager region, the sphere bound for the AWLN channel is thus derived, where the radius of the sphere is optimized to tighten the bound. 
Furthermore, the asymptotic IS gain of the proposed IS estimator implemented on the AWGGN channel is derived in a multiple integral form. Specifically for the AWLN and AWGN channels, the IS gains can be derived in a one-dimensional integral form based on the closed-form conditional PEPs. Simulation results have shown the efficiency of the proposed IS estimator under different channel parameters and coding schemes. The accuracy of the derived IS gain in predicting the performance of the estimator is also verified.

\appendix[Proof of Theorem \ref{th:sb_l1}]
\label{appendix:sb_l1}
Denote $X_{d_0} = \sum_{i=1}^{d_0} Z_i$ as a random variable. Its two-sided Laplace transform is
\begin{align}
	F_{X_{d_0}}(s) &= E\left[e^{-sX_{d_0}}\right] = \underbrace{\int \cdots \int}_{d_0} e^{-s\sum_{i=1}^{d_0}z_i} f(z_1,\cdots,z_{d_0}, D|r) dz_1\cdots dz_{d_0} \nonumber \\
	&=\frac{\Gamma(n)r^{1-n}}{2^d \Gamma(n-d_0)}  \underbrace{\int_0^2 \cdots \int_0^2}_{d_0}  e^{-s\sum_{i=1}^{d_0}z_i} \left(r-2d_2 - \sum_{i=1}^{d_0}z_i\right)^{n-d_0-1} \hspace{-3pt} dz_1\cdots dz_{d_0}. 
	\label{eq:l1_laplace}
\end{align}
Define
\begin{align}
	\xi_k(a) \triangleq \int_0^2 e^{-sx}(a-x)^k dx = \sum_{m=0}^k \left(-\frac{1}{s}\right)^{k+1-m} \frac{k!}{m!} \left(e^{-2s}(a-2)^m-a^m \right). \nonumber
\end{align}
Then the following double integral can be expressed in a closed-form as
\begin{align*}
	&\int_{0}^2\int_{0}^2 e^{-s(z_1+z_2)} \left(r-2d_2-\sum_{i=1}^{d_0}z_i \right)^{n-d_0-1} dz_1 dz_2 \\
	= & \int_{0}^2 e^{-sz_2} \xi_{n-d_0-1}\left(r-2d_2-\sum_{i=2}^{d_0}z_i\right) dz_2 \\
	= & \sum_{m=0}^{n-d_0-1} \left(-\frac{1}{s}\right)^{n-d_0-m} \frac{(n-d_0-1)!}{m!} \left(e^{-2s}\xi_{m}\left(a-2\right) - \xi_{m}\left(a\right) \right) \\
	= & \sum_{m=0}^{n-d_0-1} \left(-\frac{1}{s}\right)^{n-d_0-m} \frac{(n-d_0-1)!}{m!}\sum_{q=0}^{m} \left(-\frac{1}{s}\right)^{m+1-q}\frac{m!}{q!} \sum_{l = 0}^2 \binom{2}{l}(-1)^{2-l}e^{-2ls}(a-2 l)^q \\
	= & \sum_{m=0}^{n-d_0-1} \sum_{q=0}^m \left(-\frac{1}{s}\right)^{n-d_0+1-q} \frac{(n-d_0-1)!}{q!} \sum_{l = 0}^2 \binom{2}{l}(-1)^{2-l}e^{-2ls}(a-2 l)^q \\
	= & \sum_{q=0}^{n-d_0-1} (n-d_0-q) \left(-\frac{1}{s}\right)^{n-d_0+1-q} \frac{(n-d_0-1)!}{q!} \sum_{l = 0}^2 \binom{2}{l}(-1)^{2-l}e^{-2ls}(a-2 l)^q,
\end{align*}
where $a = r-2d_2-\sum_{i=3}^{d_0}z_i$ for simplicity.

Consequently, the multiple integral in (\ref{eq:l1_laplace}) can be solved in a similar manner. A tedious computation yields
\begin{align}
	F_{X_{d_0}}(s) &=\frac{\Gamma(n)r^{1-n}}{2^d \Gamma(n-d_0)} \sum_{m=0}^{n-d_0-1} \binom{n-2-m}{d_0-1}\left(- \frac{1}{s}\right)^{n-1-m} \frac{(n-d_0-1)!}{m!}  \nonumber\\ 
	& \hspace{180pt} \cdot\sum_{l=0}^{d_0} \binom{d_0}{l} (-1)^{d_0-l}e^{-2ls}(r-2d_2-2l)^m. 
	\label{eq:l1_laplace2}
\end{align}
The p.d.f. of $X_{d_0}$ can be derived by the inverse Laplace transform,
\begin{equation*}
	f_{X_{d_0}}(x) = \frac{1}{2\pi j} \int_{c-j\infty}^{c+j\infty}e^{sx}F_{X_{d_0}}(s) ds.
\end{equation*}
We can evaluate the above integral using contour integration, the residue theorem, and Jordan's lemma. After simplification, it becomes
\begin{align*}
	&f_{X_{d_0}}(x) = \frac{\Gamma(n)r^{1-n}}{\Gamma(d_0)\Gamma(n-d_0)2^d} \sum_{m=0}^{n-d_0-1}\binom{n-d_0-1}{m} \sum_{l=0}^{d_0} \binom{d_0}{l} (-1)^{n+d_0-1-m-l} \\
	&\hspace{130pt} \cdot (r-2d_2-2l)^m(x-2l)^{n-2-m} H(x-2l), \quad 0 \le x \le r-2d_2.
\end{align*}

Then, the conditional PEP in (\ref{eq:condition_PEP_p1}) can be derived as 
\begin{align*}
	&\text{Pr}(\Delta_d<0 |\|\mathbf{Z}\|_1 \hspace{-3pt} =r) = \hspace{-3pt} \sum_{d_0 = 1}^d \sum_{d_2 = 0}^{d-d_0} \binom{d}{d_0} \hspace{-2pt} \binom{d-d_0}{d_2} \hspace{-4pt} \int_{\underline{x}}^{\overline{x}} \hspace{-6pt} f_{X_{d_0}}(x) dx  + \frac{1}{2^d} \hspace{-2pt} \sum_{d_2=0}^{d} \hspace{-2pt} \binom{d}{d_2} \hspace{-3pt} \left(\frac{\overline{x}}{r}\right)^{n-1} \hspace{-12pt} H(\overline{x}) H(-\underline{x}) \\
	&= \frac{1}{2^d}\sum_{d_0 = 1}^d \binom{d}{d_0}  \frac{\Gamma(n) r^{1-n}}{\Gamma(d_0)\Gamma(n-d_0)} \sum_{d_2 = 0}^{d-d_0} \binom{d-d_0}{d_2} \hspace{-5pt} \sum_{m=0}^{n-d_0-1} \hspace{-5pt} \binom{n-d_0-1}{m} \sum_{l=0}^{d_0} \binom{d_0}{l} \frac{(-1)^{n+d_0-1-m-l}}{n-1-m} \\
	&\hspace{15pt} \cdot (\overline{x}-2l)^{n-1} \left( H(\overline{x}-2l) - \left( \frac{\underline{x}-2l}{\overline{x}-2l} \right)^{n-1-m} \hspace{-15pt} H(\underline{x}-2l)\right) +\frac{1}{2^d}\sum_{d_2=0}^{d} \binom{d}{d_2} \left(\frac{\overline{x}}{r}\right)^{n-1} \hspace{-12pt} H(\overline{x}) H(-\underline{x}),
\end{align*}
where $\overline{x} = r-2d_2$ and $\underline{x} = d-2d_2$.


%




\ifCLASSOPTIONcaptionsoff
  \newpage
\fi



\bibliographystyle{IEEEtran}
\bibliography{bare_jrnl}

%

%

%





\end{document}